\documentclass[12pt]{article}
\usepackage[dvips]{graphicx}

\usepackage{color}
\usepackage[x11names,rgb]{xcolor}

\usepackage{amsfonts}
\usepackage{amsmath}
\usepackage{amssymb}
\usepackage{amsthm}
\usepackage{enumerate}

\usepackage{natbib}
\bibpunct{(}{)}{;}{a}{,}{,}

\newtheorem{proposition}{{\bf \sc Proposition}}

\newtheorem{definition}{{\bf \sc Definition}}

\newtheorem{claim}{{\bf \sc Claim}}
\newtheorem{assumption}{{\bf \sc Assumption}}
\theoremstyle{remark} 
\def\eproof{\qed}

{\begin{Sbox}\begin{minipage}}%
{\end{minipage}\end{Sbox}\fbox{\TheSbox}}

\begin{document}

\title{Bring your friend! \\ Real or virtual?\thanks{%
We wish to thank  Larry Samuelson for his precious suggestions.   We also acknowledge Gani Aldashev, Paul Belleflamme, Jean Marie Baland, Ennio Bilancini, Timoteo Carletti, Jacques Cr\'{e}mer, Marc Bourreau, Rosa Branca Esteves, Mathias Hungerbuhler,  Stefano Galavotti, Edoardo Grillo, Mark Le Quement, K\'{a}roly Tak\'{a}cs and Eric Toulemonde for their advices. We wish to thank the audience to the seminars in  Central European University, Cergy-Pontoise, Corvinus University, FUSL, Modena, MTA TK ``L\"{e}ndulet'' RECENS, Sassari, University  of East Anglia, Bologna, IMT-Lucca,
and the the participants at the AFSE 2014, ASSET 2013 and ASSET 2014 Conferences as well as those at Ecore Summer School 2013 (Leuven), Bomopav 2015 (Modena),  IO in the Digital Economy Workshops (Li\`{e}ge) and  GRASS 2016  for their useful comments and critics. Elias Carroni acknowledges the  ``Programma Master \& Back - Regione Autonoma della Sardegna" and the Labex MMD-II for financial support. Simone Righi acknowledges the ``International Mobility Fund",  ``International Publication Fund" and the  "Lend\"{u}let" Program of the Hungarian Academy of Sciences as well as the Hungarian Scientific Research Fund (OTKA K 112929) for financial support.
P.\ Pin acknowledges funding from the Italian Ministry of Education Progetti di Rilevante Interesse Nazionale (PRIN) grant 2015592CTH.}}

\author{Elias Carroni\thanks{%
Department of  Economics (DSE), Alma Mater Studiorum -Universit\`a di Bologna, Italy.
Email: \url{elias.carroni@unibo.it}.}
\and
Paolo Pin\thanks{%
Department of Decision Sciences, IGIER and BIDSA, Bocconi, Italy. Email: \url{paolo.pin@unibocconi.it}.}
\and 
Simone Righi\thanks{Department of Agricultural and Food Sciences  (DISTAL), Alma Mater Studiorum - Universit\`a di Bologna, Italy \& 
MTA TK "Lend\"{u}let" Research Center for Educational and Network Studies (RECENS),
Hungarian Academy of Sciences.
 Email: \url{s.righi@unibo.it}.}
 }

\maketitle

%

\newpage

\begin{abstract}
A monopolist faces a partially uninformed population of consumers, interconnected through a directed social network. In the network, the monopolist offers rewards to informed consumers (influencers) conditional on informing uninformed consumers (influenced).  Rewards are needed to bear a communication cost. 
We investigate the incentives for the monopolist to move to a denser network and the impact of this decision on social welfare. 
Social welfare  increases in information diffusion which, for given communication incentives, is higher in denser networks. However,  the monopolist internalizes transfers and thus may prefer an environment with less competition between informed consumers. The presence of highly connected influencers (hubs) is the main driver that aligns monopolist incentives and  welfare. 
\end{abstract}

\noindent \textbf{JEL\ classification}: D42,  D83, D85.


\section{Introduction}
\label{sec:intro}

Programs that attribute referral bonuses to customers are an established marketing strategy through which companies attempt to increase their market penetration. This strategy is effective since consumers are part of a network of acquaintances and thus can be incentivized to use their social relationships to diffuse the knowledge about the existence of the company's product. 

In the typical referral--bonus program, the company offers rewards to its established customers, provided that they are able to convince some peers to become new clients. In order to obtain rewards, old customers need to invest in their existing social network by informing their peers about the existence of the product. Depending on their willingness to pay, newly--informed agents will then decide upon purchase. 

Referral bonuses are generally used in markets for subscription goods and services. 
In the market for online storage services, Dropbox offers free storage space to clients that convince their friends to subscribe. According to \cite{huston}, founder and CEO of Dropbox, the referral program extended their client basis by 60\% in 2009 and referrals were responsible for 35\% of new daily signups.
Moreover, banks offer advantageous conditions to old customers introducing new ones. Better conditions are provided both in the form of higher interest rates on the deposit and in that of lowered service fees. Another form of reward used by banks is to embed the rewarding mechanism in established customer loyalization schemes awarding points in exchange for referrals. Such points can then be used to claim prizes (mobile phones, televisions etc).
Other well--known examples can be found in markets for massively multiplayer on--line games, payment systems, touristic accommodation, online content providers,  food--service sector and local shops.

Referral--marketing strategies have been used by companies since long time ago on their local markets. However, modern information technologies dramatically reduce the cost of activating one's own social network, thus incentivizing the use of such strategies on mass markets and, in turns, strongly increasing their potential reach. Thanks to modern communication technologies such as instant messaging, chats, e--mails and online social networks (OSNs), informing the whole own ego--network turns out to be very cheap and quite independent of its size. Indeed, digital communication among consumers has a very low marginal cost and, independently of the means (sms, emails or OSNs posts), entails a fixed cost of communication. Namely, one or  few clicks suffice to inform all acquaintances.

However, diffusing information through OSNs or through the real--world network is  different from the seller's viewpoint, who can  further decide  among many channels. 
 Therefore, her problem is to decide whether to implement a referral program in the real  or in the virtual world and how many channels to use for information  diffusion. 
 Indeed,  the world is such that individuals have  ``real'' and ``virtual'' friends and  the network offline is a subset of the number of peers that are met in the OSNs such as Facebook and  Twitter  (\citealt{vitak2011ties} and \citealt{pollet2011use} among others).\footnote{To put it differently, people may also have different friends in different social networks and the total ego-network in Facebook plus Twitter is a superset of the Facebook-only or Twitter-only ego-network.}
As a matter of fact, examples of both choices can be found. Within the market of touristic accommodation, hotels stimulate word--of--mouth (WOM) in offline social networks, whereas  online touristic accommodation marketplaces such us AirBnB favor WOM happening through OSNs. A similar heterogeneity  of strategies can also be found  in the banking and in the food--catering sectors. 

In all these examples the company provides incentives that target popular and informed  clients, proposing them to sustain a costly investment with an uncertain return. 
From the point of view of the informed customer, the uncertainty of the investment in social network follows from two considerations. Firstly, some of the peers contacted may not be willing to buy the product even once aware of its existence. Secondly, uninformed consumers may get information about the service from multiple sources, while in most referral programs only one person can receive the resulting bonus. We take into account both these issues when modelling the expectations of customers considering to activate their social network. 

The current literature mainly focuses on pricing in OSNs (\citealt{bloch2011pricing}) and telecommunication services (\citealt{shi2003social}), where sellers can directly observe the precise structure of the consumers' network. Referral programs are a simpler but widely used strategy which only requires very limited information about consumers' social interactions. Indeed, following the approach of  \cite{fainmesser2015pricing},  in our setup the seller is only aware of the distribution of the influenced  (in--degrees) and influencing ties (out--degrees). The out--degree matters for the probability that informed people pass the information if adequately incentivized. The in--degree  matters for the probability that uninformed individuals receive the information.
Under limited information, the company cannot directly price--discriminate according to the precise position of the consumer in the social network. However, it can still influence clients' decisions by setting prices and bonuses so that some of them will  diffuse information. 

The power of referrals follows from the fact that each player in the market has incentives that favor the success of this strategy. The producer wants to extend her client base, old customers are motivated by expected rewards and  potential new buyers are given the opportunity to learn about the existence of a potentially valuable service. 
Consequently, the structure of incentives of referral--marketing strategies makes them an effective tool in the presence of significant informational problems on the consumers' side, i.e., when some of the potential customers are unaware of the existence of the product. This situation is typical when a product or service is relatively new on a market and when the existence of many specialized markets leads the consumers to a situation of information overload (\citealt{Zandt2004}). Mass media advertisement can provide a partial solution to this informational problem. However, it is well--known (\citealt{lazarsfeld1955personal}) that information coming from the mass--media is not fully trusted by consumers, who tend to be more influenced by social neighbors. 
The strategy we study is thus an effective and relatively cheap alternative for companies to expand their client base as it allows to harness the power of customers' social networks.

In this paper, a population of  consumers is divided between individuals informed about the existence of the product and others who are not, with a directed network linking the former to the latter. If an informed consumer buys the product, then she   can   inform all her out--neighbors facing a lumpsum cost of communication.  A monopolistic seller decides upon the introduction of a referral program.  She sets prices in two periods and, after the first one, she offers a unitary prize  for each new (previously uninformed) client convinced to buy.
 Each consumer has private information about his own characteristics (degree and preference over the product), whereas the distribution of these characteristics is common knowledge.
 
 Our setup allows  for findings along different lines. First,  we are able to highlight the main tradeoffs faced by the monopolist for given network distribution. Second, we study monopolist's incentives of choosing an online vs offline referral program. Third, we highlight the potential  misalignment of the monopolist's incentives with respect to the social welfare.
 
Consumer's decision on communication as well as the corresponding gain depend on her popularity (the out--degree): only sufficiently connected individuals are incentivized to pay the communication cost. From the point of view of the monopolist, the optimal reward  trades--off between unitary margins and the provision of adequate communication incentives. Reward's size  depends on the distribution of degrees.  The out--degree is crucial to determine the number of consumers who pass the information and the expected bonus that each of them   receives.  For given incentives provided to informed buyers, the in--degree  drives instead  the power of the referral program in terms of informational spread.

The assessment of the  impact on monopolist's profits and welfare of moving from an offline to an online strategy is provided by using the concept of first--order stochastic dominance (FOSD), applied to the degree distribution. The general solution is non--trivial, and the objective of the monopolist may be misaligned with welfare because she does internalize the size of transfers.

We provide results on this possible misalignment by taking broad  classes of degree distributions.
 When both in-- and out--degrees are homogeneous, we show that  a monopolist has always incentives to move online, thus  proposing the bonus to a denser network, provided that the cost of communication is sufficiently small. Oppositely, if only the in--degrees are homogenous and/or the cost of information is sufficiently high, then the monopolist  may prefer  to offer the bonus to a sparser network. 
  
When analyzing welfare, one has to take into account only the intensity of the informational spread, since the bonus is simply a transfer  that does not affect social welfare. For these reasons, if the cost of communication is sufficiently high the monopolist never maximizes welfare, which would require to induce all informed consumers to pass the information. This misalignment of monopolist's incentives has clearly an impact on the desirability of  selling the good either online or offline. The monopolist's incentives to go online turn out to be aligned with welfare when both informed and non--informed consumers are homogeneous, i.e., moving virtual is both preferred by the monopolist and welfare enhancing. Oppositely, when in--degrees are homogeneous and out--degrees not necessarily, if the cost of communication is sufficiently high the monopolist prefers selling the good offline whereas going online would be better from the social viewpoint. 
Moreover, in a more general case, we also provide sufficient conditions on the FOSD shift for the existence of a case in which the monopolist optimally moves online with  detrimental effects on welfare.

Finally, we  study numerically the marginal effects of moving online for two broad classes of in-- and out--degree distributions: random networks (\citealt{erdHos1959random}) and scale--free networks (\citealt{barabasi1999emergence}). We find that incentives to move online are  not always aligned with welfare  and this depends on the presence of hubs in the in-- and out--degree distribution. These hubs make the monopolist more prone to care about information diffusion and thus welfare.

The remaining part of this paper is divided as follows. After discussing the related literature in the following section, we outline the mathematical aspects of the model and we provide the equilibrium in Section \ref{model}. Then, we follow up by performing some comparative statics analysis on the effects of moving online on welfare and on seller's profits (Section \ref{virtual}), and  we finally draw the conclusions (Section \ref{conclusions}).

\section{Related literature}
\label{literature} 

 The empirical literature has long considered word--of--mouth communication (WOM). The seminal work of \cite{lazarsfeld1955personal} formulated the general theory that when people speak with each other and are exposed to information from media, their decisions are based on what peers say rather than on what media communicate. They showed that an effective way for companies and governments to reach their goals is to influence a small minority of opinion leaders, who then tend to spread the message. \cite{arndt1967role} is among the first scholars to study empirically the short--term sales effects of product--related conversations, showing that favorable comments lead to an increase in the adoption of new products and vice--versa.  \cite{van2007new} and \cite{iyengar2011opinion} point out the importance of opinion leaders or influential agents in the diffusion of product adoptions. 
Our paper essentially is in concord with this empirical observation as the company targets relatively more popular agents in order to increase profits. 
Our setup enriches these results by allowing for a comparative statics analysis on networks densities. 

The solipsistic view of the consumer, which characterized the economic discipline in the past, can be relaxed considering the single agent as a member of a social group. Indeed, individuals influence and are influenced by social behavior through local interactions.
 The concept of network has been introduced and applied in a variety of fields.
As pointed out in \cite{jackson2016economic} networks influence agents' economic behavior in fields such as decentralized financial markets, labor markets, criminal behavior and the spread of information and diseases.

 There is a growing literature on learning and diffusion in social networks, summarized in the very recent surveys of  \cite{diff2016handbook} and \cite{learn2016handbook}.
 Diffusion of behaviours in a social network is known to depend on the connectivity distribution of the latter  and on specific features of the diffusion rule (\citealt{lopez2008diffusion}). Furthermore, the role of social influence on binary decisions for generic influence rules is studies by \cite{lopez2012influence}.

In recent years, several scholars have been interested by  the fact that consumers discuss with peers about the products they buy and that this can be taken as given by  sellers when defining their strategies, as noted in  the survey in \cite{bloch2016handbook}. 
There are models that accounts for network externalities between consumers, like \cite{candogan2012optimal}, \cite{fainmesser2015pricing} and \cite{Zenou}.
In this paper we focus instead on word--of--mouth (WOM) communication.
Along these lines \cite{campbell2013word} studies optimal pricing  when few consumers are initially informed and engage in WOM; \cite{galeotti2009influencing} discuss the optimal target to maximize market penetration with WOM; and \cite{galeotti2010talking} investigates the relationship between interpersonal communication and consumer investments in search. When WOM is taken as given, the key issue for the seller is to assess how the consumers' network reacts to any marketing strategy in terms of percolation of information. Instead of focusing on WOM \emph{per se}, our paper discusses communication resulting from a deliberate incentive scheme predisposed by the monopolist. In other words, the strategy we analyze generates communication  which would not exist otherwise. 

Incentivized WOM through referrals has been recently studied by \cite{leduc2015pricing} and \cite{lobel2016customer}.   \cite{leduc2015pricing} focus on the role of referrals in launching of a new product or technology of uncertain quality. They show that  high--degree people are worth being incentivized because, even though they have more incentive to free--ride, they are powerful means of information diffusion. Depending on the network structure, referral strategies can be preferred to inter--temporal price discrimination. We provide a similar result in a context where the  informational asymmetry  is on the  existence of the product. 
\cite{lobel2016customer}   study instead the optimal referral schemes in a rooted network in which a focal consumer decides whether to purchase and whether to inform friends in exchange  for a referral bonus. We share with this paper the strategy studied, but our interest rests on the business decision about  the diffusion channels assuming a given communication--incentive scheme.

 Finally,  industrial economics shifted from the network externalities approach, following the tradition of \cite{katz1985network}, to a new focus on the direct study of the effects of social interactions on the behavior of economic agents. The new tendency comprises the consideration of a subset of neighbors rather than the population overall as the main driving force influencing individual choices (\citealt{Sundararajan,dutta}). The concept of network locality has been used by \cite{dutta} to show the emergence of local monopolies with homogeneous firms competing in prices, by \cite{candogan2012optimal} to study optimal pricing in networks with positive network externalities, by  \cite{bloch2011pricing} to study the optimal monopoly pricing in on--line social networks and by \cite{shi2003social} to study pricing in the presence of weak and strong ties in telecommunication markets. The main concern of the last two papers is price discrimination based on network centrality and on the strength of social ties respectively. While their models assume a full knowledge of the links among consumers, the strategy we discuss requires only very limited information. Instead of gathering detailed information about individuals in order to directly price discriminate, the company offers incentives that motivate buyers to become channels of information diffusion.

\section{Model}
\label{model}

We consider a measurable mass of consumers.
A fraction $1-\beta$ of them is informed about the product offered by a monopolist.
The remaining fraction $\beta$ is not informed.
All consumers have a reservation value which is uniformly distributed as $U(0,1)$, and they are satiated at one unit. The uniform distribution is chosen to simplify computation and exposition, but the qualitative results would not change with any well--defined distribution function, as proved in Appendix \ref{GENERALf}.
The reservation price is assumed to be uncorrelated with being or not being  informed.

Informed consumers are linked with uninformed ones through a bipartite directed network which is the only sub--network relevant for the problem at hand. While obviously any individual  can in principle have other connections, only links going from the informed  to the uniformed group can influence people's behavior.
If an informed consumer buys the product and pays a fixed homogeneous lumpsum cost $c>0$ for \emph{passing the information}, then all her  out--neighbors become informed. \\
This network  is such that informed consumers have an out--degree $k$ distributed according to an i.i.d.~degree distribution $f(k)$, while uninformed consumers  have an in--degree $k$ distributed according to a  i.i.d.~degree distribution   $g(k)$. Nodal characteristics are private information whereas the distributions of degrees, information and reservation values are common knowledge. 
For the sake of consistency of the model,   out-- and in--degrees need to match, which is equivalent to say that:
\begin{equation}
\label{eq:consistency}
(1-\beta) \sum_{k=1}^{\infty} f(k) \cdot k = \beta \sum_{k=1}^{\infty} g(k) \cdot k \ \
\end{equation}
implying that the average in--degree (weighted  by the relative proportion of informed and uninformed  consumers) equates the average out--degree.
We define as $k_{f \min}$, $k_{f \max}$, $k_{g \min}$, $k_{g \max}$  the minimal and maximal non--null elements in the support of $f(k)$ and $g(k)$.
In the remainder of this section we keep the network fixed and solve the problem of the monopolist.

The game implied by this setup  is composed by two stages. In the first one, the monopolist sets price $p_{1}$ and informed consumers decide upon purchase. In the second stage, the latter are offered a unitary  bonus $b$ for each referred uninformed consumer who buys the product.\footnote{Our focus here is the study of the decision to go online or not. In principle, other referral bonus structures could be deviced. These are studied by \cite{lobel2016customer}, fixing the price of the product and focusing on optimal referral schemes. } Uninformed consumers pay price $p_2$. If $\ell$ informed  consumers pass the information to the same uninformed buyer, and that uninformed buyer buys the product,  each of them gets $\frac{b}{\ell}$ in expected terms.

So, the problem of the monopolist is to maximize, with $p_1,p_2 \in [0,1/2]$, and $b \in [0,p_2]$, the following objective function 
\begin{equation}
\label{eq:monopolist_problem}
\pi(p_1,p_2,b)  = 
(1-\beta)  p_1 (1-p_1) + \beta  (p_2-b) (1-p_2) \Gamma(L),
\end{equation}
where $L$ is the fraction of links from informed to uninformed consumers through which the information is passed, while $\Gamma(L)$ is the fraction of uninformed consumers who receive the information. For each uninformed consumer reached by information and buying the product, one bonus is paid to an informed consumer. Hence this must be subtracted from  the price $p_{2} $ and acts as a marginal cost in the second term of  Equation \eqref{eq:monopolist_problem}.

 Given the lumpsum cost of communication and the bonus $b$, only some informed consumers  will pass the information.
By now let us call $L$ the fraction of links from informed  to uninformed consumers in which word--of--mouth communication is used.
From this follows that the expected fraction of uninformed consumers who get the information is:
\begin{equation}
\label{eq:Gamma}
\Gamma(L) = \sum_{k=1}^{\infty} g(k) \left(  1-(1-L)^k \right) = 1 - \sum_{k=1}^{\infty} g(k) (1-L)^k \ \ .
\end{equation}

It is important to notice  that $\Gamma(L)$ is an increasing and convex function of $L$,\footnote{Notice that, if $g(1)<1$, then $\Gamma(L)$ is increasing and  concave in $L$, because each element $(1-L)^k$, for $k\geq2$, is decreasing and concave in $L$. Note also that, whenever $g(1)<1$ and $0 \leq L \leq 1$, we have $\Gamma(L)>L$.
$\Gamma(L)$ is actually strictly concave whenever $g(1)<1$, which is the non--trivial case of at least some receivers with multiple links.}
 with $\Gamma(0) =0$ and $\Gamma(1)=1$, so that for interior values of $L$ we have $\Gamma(L)>L$.
This implies that the fraction of receivers of the word--of--mouth communication is always greater than the fraction of active communication channels, and this has a positive  effect on information spread. Moreover, as long as the receivers have multiple in--links, the offer of the bonus also induces some competition between informed consumers. Indeed, on a social network, an uninformed individual can become aware of the product by multiple sources. Therefore, the expected value of an informed consumer when passing the information needs to take into consideration this congestion effect. 
Let us call as $\phi (L,p_2)$ the expected value (rescaled by $b$) of one out--link when passing the information.
$\phi (L,p_2)$ is given by (i)  the probability $1-p_2$ that the recipient of the message will buy the product, multiplied by  (ii) the expected bonus the sender will receive, which depends on how many other informed consumers connected to the receiver may also have passed -- in expectation -- the information:

\begin{eqnarray}
\phi (L,p_2) &  = & (1-p_2)  \sum_{k=1}^{\infty} g(k)  \left( 1 - \sum_{t=0}^{k-1} {k-1 \choose t} L^t (1-L)^{k-1-t} \frac{t}{t+1} \right) \nonumber
\\
& = & (1-p_2)  \sum_{k=1}^{\infty} g(k)  \left( \frac{1-(1-L)^k}{k L} \right). 
\label{eq:phi}
\end{eqnarray}

Congestion implies that $\phi(L,p_{2})$ is a decreasing function of $L$, so that the more people pass the information, the less would be the expected value that each of them benefits. The ``crowding'' effect is
  stronger as the number of communicators gets higher, but  the problem becomes marginally less severe as many informed buyers pass the information.\footnote{ Whenever $g(1)<1$, $\phi(L,p_2)$ is decreasing and  convex in $L$,\footnote{%
Actually, for $k=1$ we have that $ \frac{1-(1-L)^k}{k L}$ is a constant, it is linear for $k=2$, and then it becomes decreasing and strictly convex for each $k \geq 3$.
So, it is strictly convex only if $g(1)+g(2)<1$.
$\phi$ is $1-p_2$ when $L=0$, and then it decreases to $(1-p_2)  \sum_{k=1}^{\infty} \frac{g(k)}{k} $, when $L=1$.}
and each element $ \frac{1-(1-L)^k}{k L}$ is decreasing and convex in $k$. }

Now, in order to close the model, we explicit how the fraction of activated links $L$. Given the communication  cost he faces, an informed consumer with out--degree $k$ will pass the information only if
\begin{equation} \label{threshold_k}
k  b \phi (L,p_2) \geq c \ \ .
\end{equation}
This means that actually there can be a lower level $\underline{k}$ for which informed consumers with that degree may be indifferent between passing or not the information. \\
Note that a mixed equilibrium with  a threshold $\underline{k}$ in which some informed consumers pass the information and others do not is actually stable, because if the fraction of those passing the information increases (respectively, decreases), then $\phi (L, p_2)$ decreases (increases), and from Equation \eqref{threshold_k} this provides an incentive for other informed consumers to stop (start)  passing the information.
We  define as $\rho$ the fraction of links, conditionally on informed consumers having bought the product, that pass the information, so that we also have: 
\begin{equation}
\label{eq:L}
L=(1-p_1)\rho \ \ .
\end{equation}
Let us define the average out--degree $E(k)=\sum_{k=1}^{f_{\max}} f(k) \cdot k$, and 
\begin{equation}
\Lambda(\rho)= \max \left\{ \underline{k} : \frac{ \sum_{k=\underline{k}}^{f_{\max}} f(k) \cdot k }{E(k) } \geq \rho \right\}
\end{equation}
as the minimal threshold on degree that is consistent with $\rho$.
Since $f(k)$ is a discrete distribution, $\Lambda(\rho)$ is discontinuous.\footnote{
However, it is easy to check that it is lower semi--continuous and that it has only a countable number of lower jumps.}
We are now ready to prove Proposition  \ref{prop:all1/2}.

\begin{proposition}
\label{prop:all1/2}
Every solution of the problem is such that $p_1^* \leq p_2^* =1/2$ and $b = \frac{c}{\underline{k}^{*} \phi (\rho^{*},p_2)}$. 
Moreover:
\begin{enumerate}[(i)]
\item if the equilibrium is such that consumers at $\underline{k}^{*}$ play different strategies, then $p_1^{*}=p_2^{*}$;
\item if otherwise, $\rho^*$ is a point of discontinuity for $\underline{k}^{*}$, then $p_1^{*}$ may be less than $p_2^{*}$.
\end{enumerate}
\end{proposition}

\begin{proof}
See \ref{proof_prop:all1/2}.
\end{proof}

Proposition \ref{prop:all1/2} tells that, once the monopolist identifies the $\rho^*$ that maximizes profits, the bonus will be just sufficient  to induce some or all of the buyers with out--degree $\underline{k}^*$ to pass the information. Once information diffusion takes place, the price offered to newly informed consumers will simply be the monopoly price $p_2^*$.   Price $p_{1}^{*}$ turns out to be always equal to $1/2$, except in points of discontinuity,  where  all people at $\underline{k}^{*}$ pass the information,  so that a marginal increase in communication incentives makes the $\underline{k}^{*}$ suddenly jump downwards. Only in those cases,  the monopolist offers some discounts, as show in the examples of Figure \ref{c_on_p1}.

 \begin{figure}[h!]
 \centering
         \includegraphics[width=.32\textwidth]{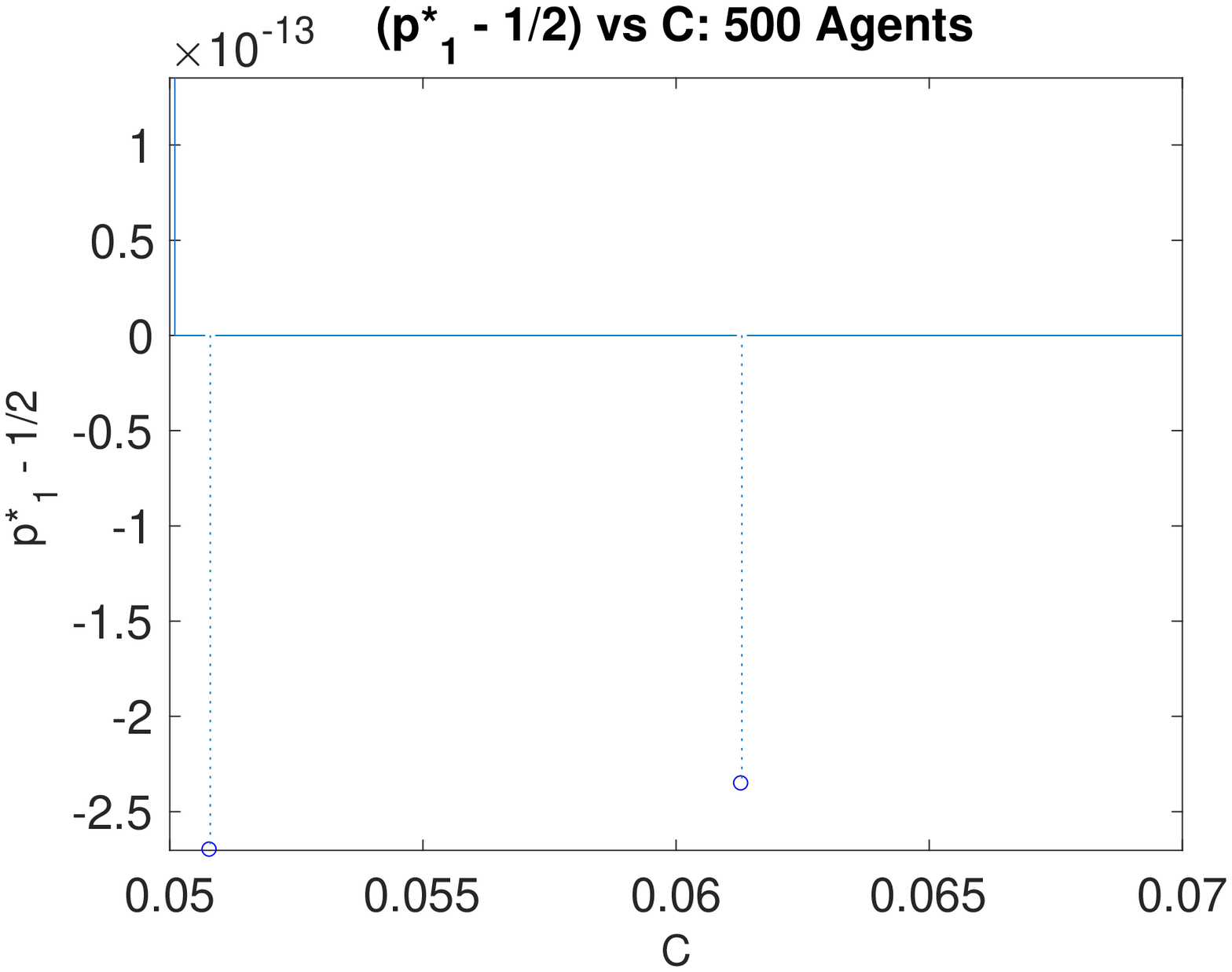}
         \includegraphics[width=.32\textwidth]{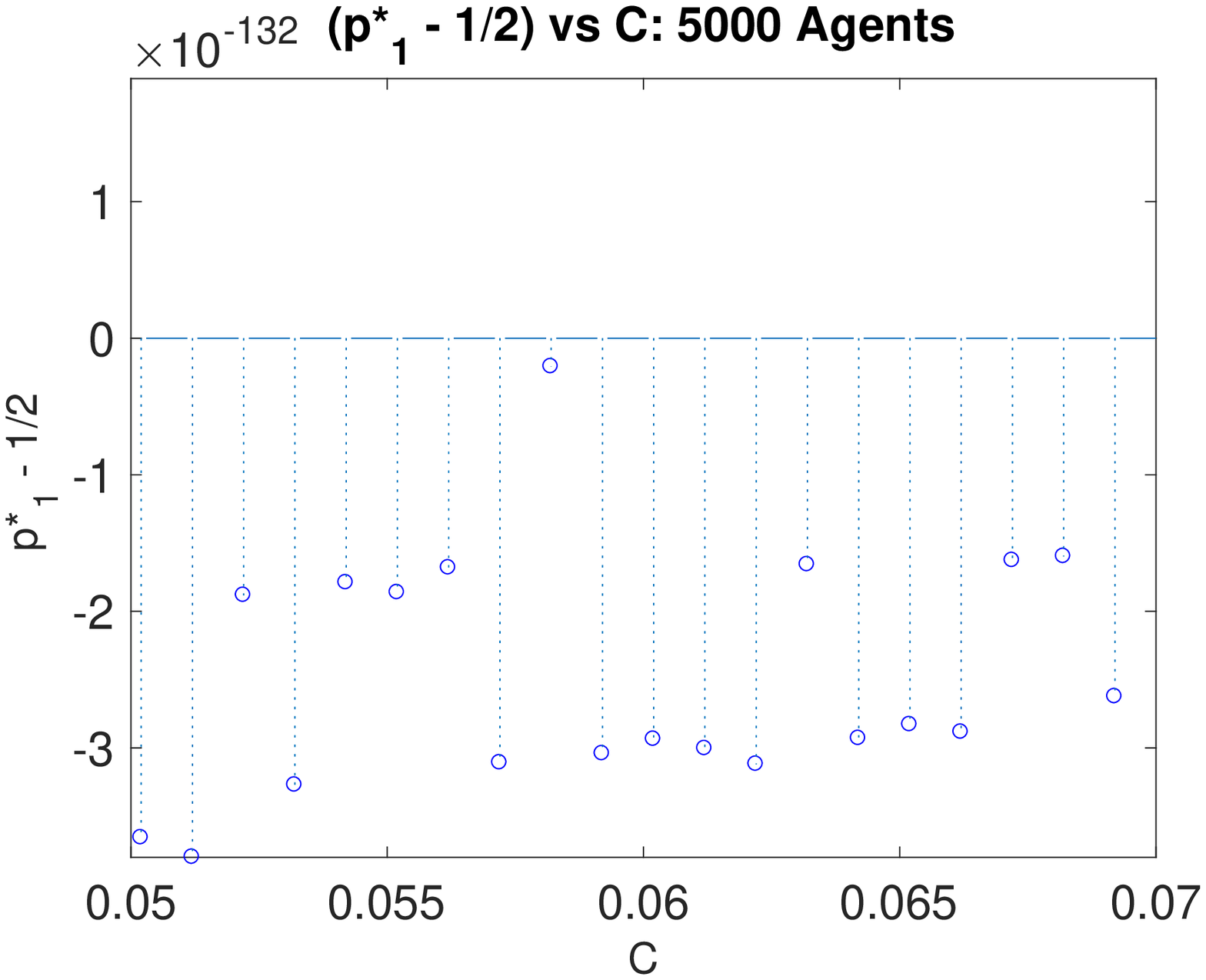}
         \caption{Effect of $c$ on $p_1^*-1/2$. Results are reported  for a random network with probability of link equal to $0.1$ and  $\beta=0.3$. Left panel: $500$ agents. Right Panel: $5000$ agents.}
 \label{c_on_p1}
\end{figure}

One may wonder why the monopolist does not always offer discounted prices to informed consumers in order to enlarge the pool of potential communicators and, consequently, the information spread. The reason is that  the incentive to do so is moderated by the resulting increase in congestion and  Proposition \ref{prop:all1/2} shows that  the two effects perfectly offset, except in points of discontinuity. Furthermore, in Appendix \ref{numerical_analysis} we show numerically that even in these cases $p_1^*$ is  extremely close to $1/2$ and the difference vanishes very fast  as the network becomes larger.\footnote{For example, in a Bernoulli random network, the difference between $p_1^*$  and $1/2$ has an order of magnitude of $10^{-13}$, when the expected degree is $50$ while decreasing to an order of magnitude of $10^{-132}$, when the expected degree is $500$.}
Therefore, in what follows, we can safely assume:  

\begin{assumption}
\label{claim}
Independently on the optimal choice of $b^*$, the maximization problem is solved by $p_1^*= p_2^*=1/2$.
\end{assumption}
 
So, from  Assumption \ref{claim}, we can restate the objective function \eqref{eq:monopolist_problem} of the monopolist as
\begin{eqnarray}
 \label{eq:monopolist_problem4}
\pi(p_1,p_2,\rho)  & = &
(1-\beta)  p_1 (1-p_1) + \beta  \left( p_2-   \frac{c}{ \underline{k} \phi ((1-p_1)\rho,p_2)  } \right) (1-p_2) \Gamma\left( (1-p_1) \rho\right) \nonumber \\
\pi(\rho )  & = &
\frac{1-\beta}{4} + \frac{\beta}{2}   \left( \frac{1}{2} -   \frac{c}{ \underline{k} \phi \left( \frac{\rho}{2}, \frac{1}{2} \right)  } \right)  \Gamma(   \rho / 2 ) \ \ 
\end{eqnarray}
with  $\rho^* \in [0,1]$. The first term of Equation \eqref{eq:monopolist_problem4} refers to informed consumers while the second one  refers to uninformed ones. In the second term, the monopolist balances the diffusion effect measured by $\Gamma(\rho / 2)$ and the congestion effect driven by $\phi\left( \frac{\rho}{2}, \frac{1}{2} \right)$. To put it differently, if the monopolist  induces a higher  fraction of informed buyers to pass the information ($\rho$), this increases demand of non informed buyers ($\Gamma$) but  decreases the net profits generated by each of them, as $\phi$ makes the required bonus bigger. 
 
\subsection{Misalignment of incentives}

The incentives of the monopolist are not necessarily aligned with  social welfare. Indeed, the social value of information is clearly greater than its value to the monopolist, who is paying for generating communication.
The aggregate welfare is given by:
\begin{eqnarray}
W & = & (1-\beta) \left( \int_{p_1}^1 (v-p_1) dv +p_1(1-p_1)\right)+ \beta \cdot  \Gamma(L) \left(\int_{p_2}^1 (v-p_2) dv +p_2(1-p_2)\right) \nonumber \\
& = &  \frac{3}{8} \left( (1-\beta)     +
 \beta \cdot   \Gamma(L) \right).
\label{eq_welfare}
\end{eqnarray}

It is clear that, fixing $p_1=1/2$, welfare is maximal when $L=1/2$, i.e., when $\underline{k}= k_{f \min}$.\footnote{%
In the definition \eqref{eq_welfare} of welfare we neglect the aggregate communication cost incurred by the originally informed agents.
It should be noted however that, even if we took it into account, it would still be the case that the optimal level of communication for the monopolist, if non maximal, would be strictly lower than the optimal level of communication for welfare.
Intuitively, that is because, at the  bliss point of the monopolist, the marginal costs would be equal for both (because the cost for making marginally more informed consumers communicate is the same), but the marginal benefit is lower for the monopolist (because the expected value for new uninformed consumers is higher than the monopolist price).}
 It means that the only concern of a social planner is to maximize diffusion, and this is guaranteed by maximizing  the number of speakers.
Differently,  the profit function of the monopolist is given by \eqref{eq:monopolist_problem4}.
%
In order to better understand the misalignment of incentives between the seller and the social welfare, let us consider  what would be the expected rate of $b$ that an informed consumer would expect to get from each link, if all informed consumers passed the information. 
This is 
\begin{equation}\label{exp_rec}
E_k [ 1/k ] = \sum_{k=1}^{\infty} g(k)/k \ \ .
\end{equation}
It is possible to establish a simple sufficient condition under which the monopolist's incentives will not be aligned with welfare and she will not maximize aggregate utility.
\begin{proposition}\label{prop:misalignment}
If $E_k [ 1/k ] < 2 \frac{c}{k_{f \min}}$, then the monopolist will not maximize welfare.
\end{proposition}
\begin{proof}
See Appendix \ref{proof_misalignment}.
\end{proof}

\bigskip

The intuition of the proof is that the monopolist's incentive is aligned with welfare only if it is possible and profitable for her to provide a $b$ that is high enough so that even informed consumers with the lowest possible out--degree will pass the information. 
Another interpretation of the result is that, if the network is so dense that the expected reciprocal of the in--degree  as defined in Equation \eqref{exp_rec}  is small enough, then the incentives of the monopolist will go against aggregate welfare.
This last consideration anticipates the comparative statics exercises that we present in next section.

\section{Moving virtual}\label{virtual}

Companies running a referral program face the decision on whether to stimulate WOM in  the real--world social network of friendships and acquaintances of consumers or in an online social network like Facebook or Twitter. The latter can be seen a superset of the former, as the OSNs bypass the time and resource constraints related to the maintenance of  real--world friendships. 
In the following, we will compare a sparser and a denser social network calling them the \emph{real} and the \emph{virtual} network, respectively; and we will also call \emph{moving online} the decision to use the virtual network instead of the real one.
This is an easy way to refer to the two cases analyzed by our comparative statics exercise, but it is clear that in applications the two cases could also refer to two networks that are both \emph{virtual}.
As an example consider the case of a monopolist that chooses between using only the social network given by Facebook, or using the union of Facebook, Twitter, and possibly other social media.

The point is that, when a seller decides to move online, she expects to face a denser network that is an overset of the existing one. Formally, adding more links to a network, both in--degree and out--degree distributions $f(k)$ and $g(k)$ receive a shift towards higher density.
Because of this, we can thus use  the concept of first-order stochastic dominance, adopting the definition below provided by \cite{diff2016handbook}.

\begin{definition}
\label{def_fosd}
A distribution $f'$ first-order stochastically dominates (FOSD) a distribution $f$  if,  for every $\hat{k}\in\{1,...,\infty\}$, and for every nondecreasing function $u: \mathbb{R} \rightarrow \mathbb{R}$, it holds that: 
$$\sum\limits_{k=1}^{\hat{k}} u(k) f(k)\geq \sum\limits_{k=1}^{\hat{k}} u(k) f'(k).$$
\end{definition}

We start our analysis from social welfare considerations, for which what matters is only the size of consumers reached by the news of the new product, because the bonus $b$ is just a transfer.
Then, we move to the analysis of the profit of the monopolist, who takes  into account the size of the bonus $b$, which depends not only on the amount of uninformed consumers that will receive news, but also on the competition between informed consumers, which may disincentivize WOM communication.
We study two cases. First -- as an intermediate step in our analysis -- we consider the case in which the  monopolist, when moving online, keeps $\rho^*$ fixed. 
Then, we relax this restriction allowing $\rho^*$ to optimally change going from the sparser to the denser network.  

\subsection{Fixed threshold for communication}

We start by analyzing how welfare and profits are affected by small changes in the in-- and out--degree distributions, if the level of $\rho^*$ is kept fixed.

\paragraph{Welfare.} 
Let us assume a shift in the in--degree distribution to some $g'(k)$ which FOS-dominates $g(k)$. In order to maintain matching between the numbers of  in-- and out--degrees as defined in Equation \eqref{eq:consistency}, and since we are considering a superset network, also a FOSD shift to some $f'(k)$ is needed. All this has a monotone effect on welfare, as reported below.

\begin{proposition}
If $\rho^*$ is kept fixed, welfare is increased when the nodes are more likely to have higher degrees (as a consequence of a FOSD shift).
\label{FOSDWelfare}
\end{proposition}

\begin{proof} 
See Appendix \ref{proof_FOSDWelfare}.
\end{proof}

The result in Proposition \ref{FOSDWelfare} follows from the fact that, for the welfare computation, the impact of higher network density on competition for bonuses disappears. Indeed, any $b$ paid to some consumer is just a transfer from the monopolist to the consumers and has no effects on the welfare. Without such competition, increasing density only allows more information to circulate and thus it increases welfare.

\paragraph{Profits.}
The effects of a FOSD increase in $f(k)$ and $g(k)$ on profits are not as unambiguous as in the case of welfare.
This is because a denser network increases $L$, and also $\Gamma (L)$, but the  effects on $b=\frac{c}{\phi \underline{k}}
$  go in the opposite direction, as the FOSD  boosts competition between informed consumers. As a result,  the monopolist will be  forced to increase $b$
\footnote{ Notice that this is analogous to an increase in the fraction  $\rho^{*}$ of past buyers that will pass the information.} to compensate for this, if she wants to maintain the same $\rho^{*}$.

 \begin{figure}[h!]
 \centering
       \includegraphics[width=.45\textwidth]{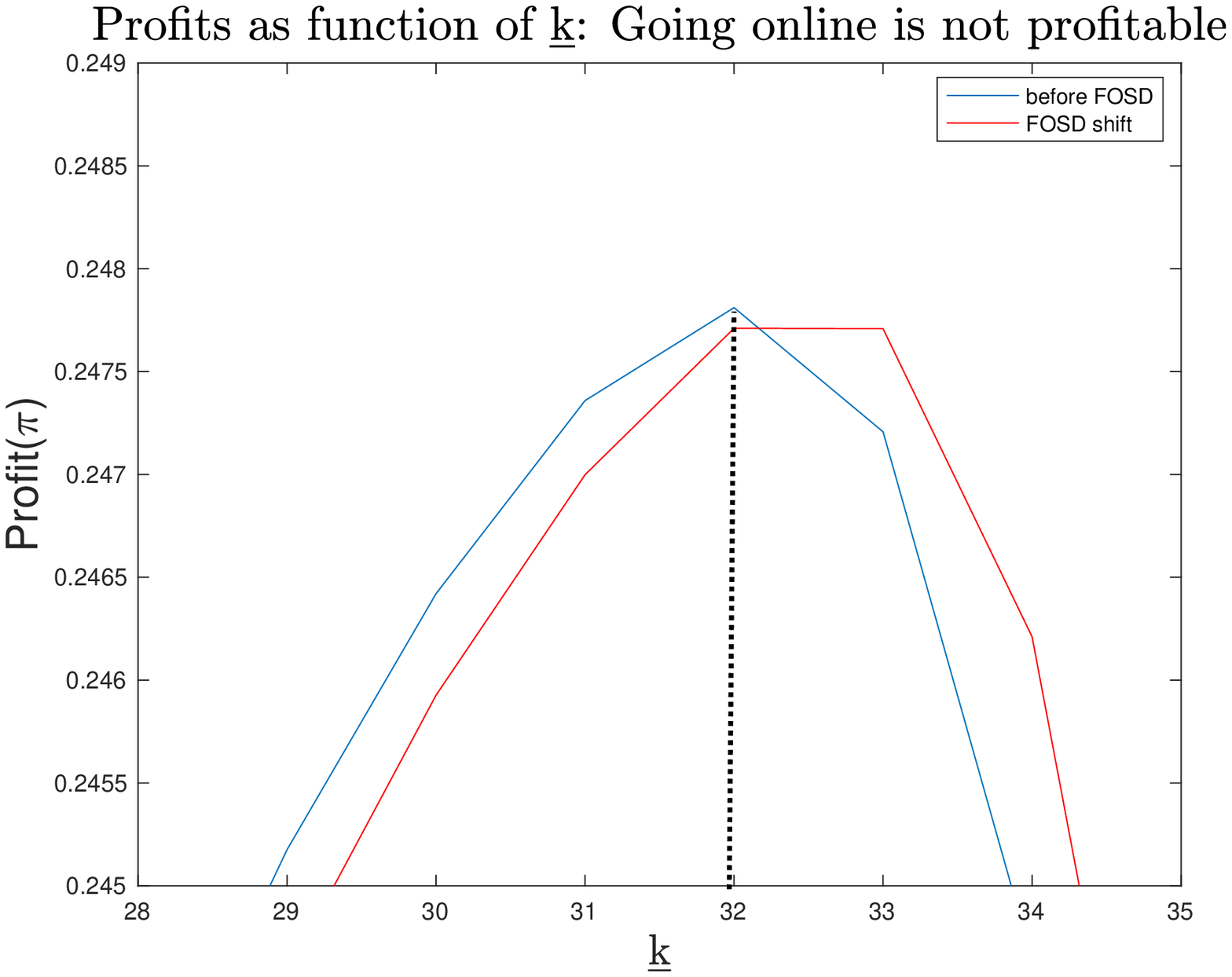}
        \includegraphics[width=.45\textwidth]{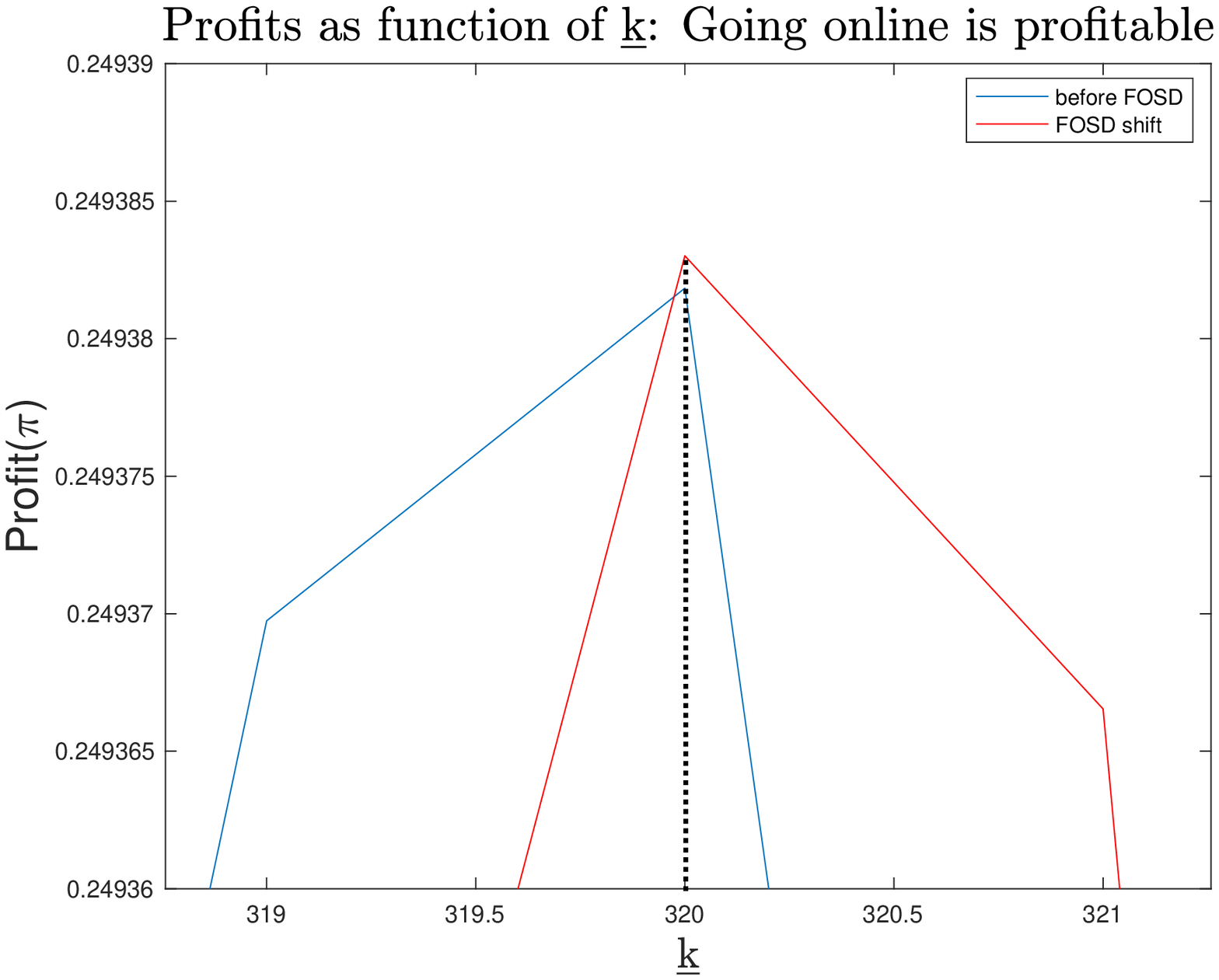}
        \caption{ Left Panel: Example in which going online is not profitable. Both in-- and out-- degree are Bernoulli random networks with probability of an in--link $\lambda_{g}=0.1$, $\lambda_{f}\sim0.0251$ and $\beta=0.2$. FOSD is done by increasing $\lambda_{g}$ by $1.6\cdot 10^{-3}$. Right Panel: Example in which going online is  profitable. Both in-- and out-- degree are random networks with $\lambda_{g}=0.44$, $\lambda_{f}\sim  0.2931$ and $\beta=0.4$. FOSD is performed  by increasing $\lambda_{g}$ by $1\cdot 10^{-3}$. In both cases network size is 1000. }
\label{examples_profit_going_online}
\end{figure}

In Figure \ref{examples_profit_going_online}, we report  examples of tiny FOSD shifts which do not entail any change in the optimal $\underline{k}$. As it can be observed, there are cases in which moving virtual is profit--enhancing  and others in which it is not. The density of the network and the proportion of uninformed people become crucial, as they determine the gains in terms of informational spread that would follow a switch to a denser network.  
The left panel considers a case in which  there are few uninformed on a relatively sparse network. In this case, a marginal increase in density creates a negligible amount of additional information spread. Moreover, a higher density also strengthens competition among informed buyers and therefore requires a higher bonus. The balance between competition and diffusion of information makes it not worth it  to go online. 
Differently, when many uninformed consumers are laid on a denser network (right panel),  information diffusion becomes more salient to the monopolist. Therefore she finds it profitable to pay the  higher  unitary bonus needed for passing the information, as the gains in terms of information diffusion are more important.

\subsection{Endogenous threshold for communication}

When setting the optimal $\rho^*$, the monopolist faces a trade--off between increasing information diffusion (represented by $\Gamma$) and reducing the cost of providing incentives. 
An increase of $\rho^*$ will boost  $\Gamma$  but will also make competition among informed consumers tougher,  requiring a higher $b$. This comparative statics analysis is non--trivial and its results strongly depend on the functional forms given to the degree distributions.
The trade--off between information diffusion and competition among consumers is clear--cut  even if not analytically provable in general. Still, the analysis of  archetypal network topologies helps us to assess the effects of the network density on monopolist incentives and optimal decisions.
 In what follows, we provide  analytical solution for networks with homogenous degree (either in-- and out--degrees  or only in--degrees) and for specific FOSD shifts. We conclude by comparing   numerical solutions on two of the most important classes of networks, namely random and scale--free networks.

%
%
%
%
%
%
%
%
%
%

\paragraph{Homogenous degree distributions.}

In what follows, we consider the simplest case in which all consumers  have the same in-- and out--degree. This allows us to get rid of any effect of increasing the density on the number of  consumers who pass the information.

\begin{proposition}\label{prop:FOSD_homogeneous}
Let the out--degree be $k_f$ and the in--degree be $k_g$ for all agents. 
In this case, moving to an online network with $k'_f>k_f$ and $k'_g>k_g$ is always weakly profitable for the monopolist. 
If the monopolist was using word--of--mouth in the sparser network, then moving online is always strictly profitable for the monopolist and welfare increasing.
\end{proposition}
\begin{proof}
See Appendix \ref{FOSD_homogeneous}.
\end{proof}

The interpretation of the results of Proposition \ref{prop:FOSD_homogeneous} is straightforward. 
 Indeed, this is a degenerate case in which the absence of variance in the degrees allows the monopolist to fully extract the surplus of communicators, who are all \emph{infra--marginal}. The resulting incentives to go online are clearcut: the choice of the monopolist is  whether making people at $k_{f}$ willing to pass the information for given network density.  Depending on the cost of communication this can be profitable or not in the sparser network. As long as the choice of going online is concerned, a denser network has a positive effect on the informational spread $\Gamma(\rho^{*})$ for given number of targeted senders $\rho^{*}$, whereas the marginal cost of providing the required incentives remains constant in both networks. Therefore,  the monopolist always  increases profits for any given number of senders. This has to be intended as a weak condition, meaning that the cost of communication might be so high to prevent profitable diffusion of information both in the sparser and in the denser network.  The circulation of more  information also guarantees  a welfare improvement. 
 However, as discussed in many cases below,  the monopolist faces a material trade--off when introducing variability in the degrees.

\paragraph{Homogenous non--informed consumers.}
In relation to Proposition \ref{prop:FOSD_homogeneous},  we now relax the assumption of homogeneity in the degree of informed consumers and we look for sufficient conditions under which the monopolist prefers to stay offline. 
We can prove the following: 
\begin{proposition}\label{prop:HOMOIN}
Consider a situation in which the in--degree distribution is homogeneous.
Consider a FOSD shift such that $k_g$ increases to $k'_g$, and 
$f(k)$ also has a FOSD shift to some $f'(k)$, so that the consistency condition in Eq. \eqref{eq:consistency} is maintained.
If $c > \frac{\beta^2}{4(1-\beta) }$, then it is possible to find a distribution $f'(k)$ such that the monopolist will obtain  strictly lower profits in the denser network.
\end{proposition}

\begin{proof}
See Appendix \ref{HOMOIN}.
\end{proof}

Differently from what discussed for Proposition \ref{prop:FOSD_homogeneous}, here the absence of variability in the in--degrees is a problem for a monopolist facing a denser network. Indeed, the congestion is, in a sense, maximal within a given network structure: the non--informed consumers are all connected to the highest number of influencers present with positive probability. Serving a denser network exacerbates the congestion effect.  Thus, if the cost of communication is sufficiently high, we can always find situations in which going online requires so high bonuses  to be profit--detrimental compared to those  of the sparser network. 

The cost's threshold increases in the proportion of non--informed consumers. As the informational problem gets more severe the monopolist is more likely, for given network, to use bonuses to generate new demand.

Finally, we need the fraction of senders not to decrease (too much) in relation to the sparser network  for the welfare to be higher when the monopolist moves online. Formally: 

\begin{proposition}\label{prop:INHOMOWELF}
In the case described in Proposition \ref{prop:HOMOIN},  $\frac{k_g}{k'_g}<
1-\frac{\log{[2-\rho']}}{\log{(2)}} $ is a sufficient condition for the FOSD shift to be welfare enhancing. \end{proposition}

\begin{proof}
See Appendix \ref{INHOMOWELF}.
\end{proof} 

The only thing that matters for welfare considerations is the number of people receiving the information, which is an increasing function of $\rho$. Proposition  \ref{prop:INHOMOWELF} expresses a sufficient condition by focusing on the limit case in which every buyer  passes the information in the original network. If the increase in density 
makes the monopolist set incentives such that all buyers pass the information also in the denser network, then welfare is trivially enhanced as condition in Proposition \ref{prop:INHOMOWELF} is always satisfied.\footnote{ Notice that when the share of people passing the information in the original network is lower than one, we can also have welfare--favorable FOSD shifts in which more people are willing to pass the information in the denser network. } 
 On the other side, if less people pass the information at equilibrium, for the change to be welfare increasing it is necessary that the improvement in information diffusion entailed by the higher density overcomes the reduction in the number of  communicators. This always happens when the FOSD shift is sufficiently pronounced, so that the condition in Proposition \ref{prop:INHOMOWELF} is satisfied. 

\paragraph{General case.}
Let us now  focus on a more general case of generic out-- and in--degree distributions, while constraining the type of FOSD shift we consider.  In this situation, we show that it is possible to construct a network denser than the original one, so that the monopolist would move to that network, but welfare would be damaged. 
In the statement we use the notation $\Lambda(\rho^*)$ to identify the fraction of informed consumers that is consistent with $\rho^*$.
Formally: 

\begin{proposition}\label{prop:general}
Suppose that $k_{f \max}$ is finite, and that the monopolist chooses $\rho^*$ such that $\Gamma(\rho^*/2)=1-\delta$, with $\delta  <  \frac{2 c}{ \Lambda(\rho^*) \phi \left( \frac{\rho^*}{2}, \frac{1}{2} \right)}$.
In this case it is always possible to consider a FOSD shift of the network, such that the monopolist would go online, but welfare would be damaged by this choice.
\end{proposition}

\begin{proof}
See Appendix \ref{proof:general}.
\end{proof}

As formally explained in the proof of Proposition \ref{prop:general}, an example of FOSD shift causing the described misalignment of incentives is a scenario in which the shift is made by adding few very influencing people (super--hubs or web--influencers) to the original network, and redistributing the additional in--links homogeneously.
If the monopolist incentivizes only these super--hubs to diffuse information, the latter would suffer very limited congestion,  making their required bonus cheap. For this to be convenient with respect to a strategy getting on board also less connected individuals, super--hubs degree must be sufficiently high in order not to loose too much in terms of information diffusion. 
However, if  the proportion of receivers among uninformed consumers in the original network was sufficiently close to $1$, this actually reduces welfare. 

This general result points to the fact that the presence of few people with tremendously high levels of influence makes  a referral strategy very attractive to the monopolist. Indeed, even a very small reward would induce these people to pass the information with little congestion. So the presence of hubs favors the going--online solution as the increase in density 
does not entail a material loss in demand. As a result, the monopolist can reach many people even  by inducing only few people to communicate. 

In order to provide the results of the present section, we  imposed  restrictions which prevented us from  assessing  the exhaustivity of our results when such conditions do not hold. We thus recur to numerical solutions to confirm the importance of hubs for the decision of the monopolist to go online and to highlight the possible incentives' misalignments.

\subsection{Numerical results} \label{simulations}

In what follows, we discuss precise numerical solutions of our model for specific classes of in-- and out--degree distributions, typical of the theoretical literature on social networks (see \citealt{random2016handbook} for a survey related to the economic literature).
The first type is the random networks (\citealt{erdHos1959random}). These graphs are characterized by a given number of nodes $n$ and a given probability $0\leq \lambda \leq 1$, which describes the chance of each link between pairs of nodes to exist. When  $\lambda$ is assumed to be equal for each pair of nodes, these networks are characterized by a binomial distribution of degrees, i.e.,  $\forall 1\leq k \leq n-1$:
\begin{equation}
h(k)=\left(\!
 \begin{array}{c}
  n-1 \\
  k
 \end{array}
 \!\right)\lambda^k (1-\lambda)^{n-1-k},
 \label{binomial}
\end{equation}
where $\lambda n$ approximates the characteristic degree of nodes in the networks. In other terms, $\lambda$ can be considered as a measure of network density. 

The second type of degree distribution we discuss characterizes networks defined as scale-free due to the tendency of the standard deviation of the degrees to diverge. This type of construction does fit many of the characteristics of empirical social networks, in particular the observation that a lot of them approximately follow a power--law degree distribution (for specific examples see \citealt{ugander2011anatomy,ebel2002scale,liljeros2001web,barabasi2002evolution,yu1965networks,albert1999internet}).
  Formally, we study networks with degree up to $n$ and degree distribution given by:
\begin{equation*}
h(k)=\frac{1/k^{\gamma}}{\sum\limits_{n\in N}(1/n^{\gamma})},
\end{equation*}
where $2< \gamma \leq 3$ represents the slope of the power law.\footnote{In between these two values the first moment of the degree distribution is finite, but the second and higher moments diverge as the network size becomes infinite. The boundaries are justified by the fact that most empirically observed networks exhibit a slope between these two values, which implies a ultra small--world network (\citealt{cohen2003scale}).  Notice that increasing the parameter $\gamma$ implies lowering the probabilities to observe highly--connected individuals, thus leading to sparser networks.}

Scale--free networks entail the presence of hubs, i.e., nodes with very high degree with respect to the networks' average, while in random networks these disproportionally--connected  nodes  are extremely rare. For this reason, in what follows we refer to the first type of networks as networks with hubs and to the latter as networks without hubs.  

Considering these network structures, we assess the impact of the FOSD shift that results from increasing $\lambda_{g}$ (for the random network) and decreasing $\gamma_{g}$ (for the scale--free networks) of $2.5\cdot 10^{-2}$, without changing $k_{f \max}$ and $k_{g \max}$. In both cases we are thus making each possible link between arbitrary agents $i$ and $j$ more likely to exist.

In sum, the numerical solutions of the model support the following claim.
\begin{claim}
Moving online becomes gradually more likely to be profit--  and welfare--enhancing as we move from a network without to a network with hubs. 
\end{claim}

 \begin{figure}[h!]
 \centering
       \includegraphics[width=.32\textwidth]{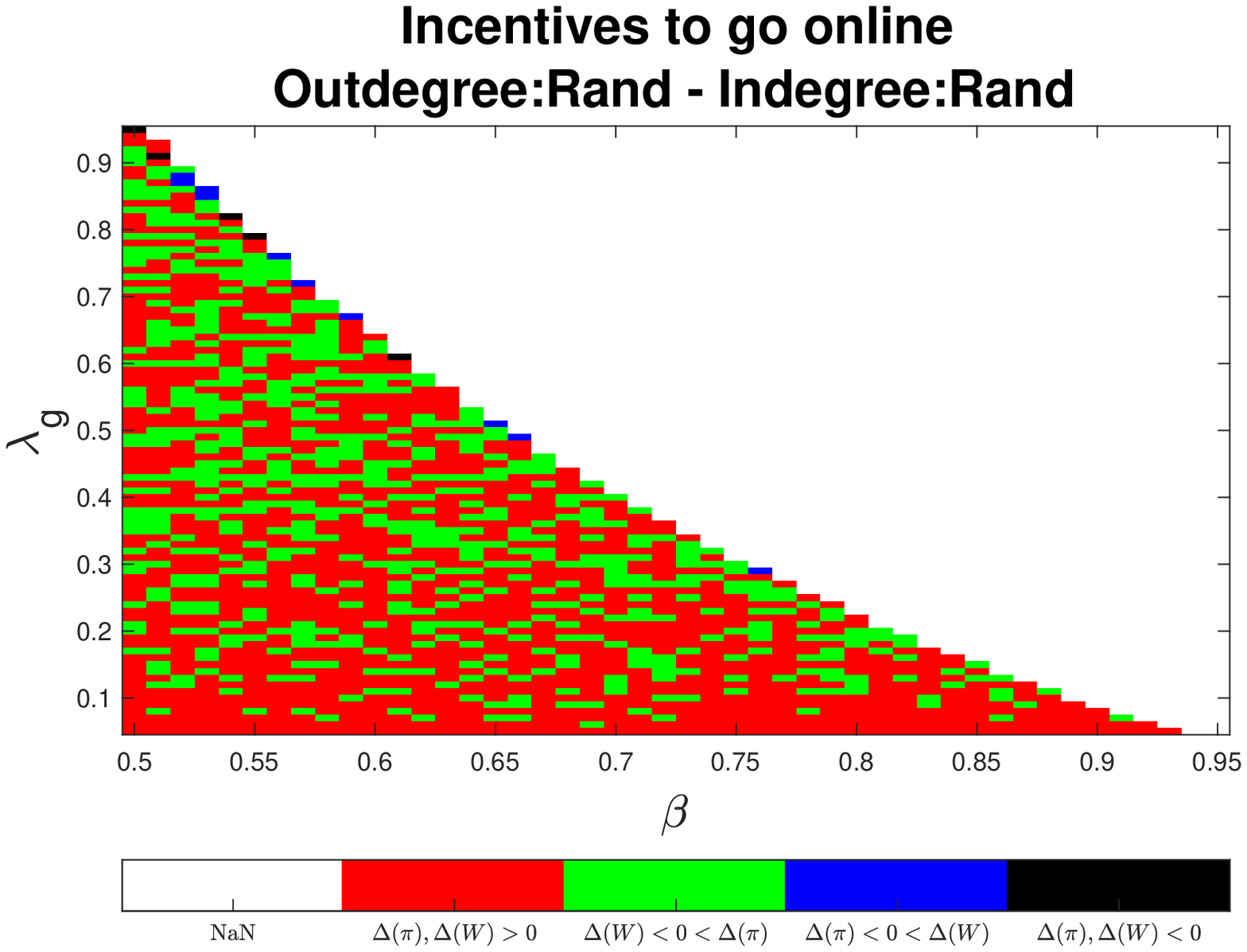}
        \includegraphics[width=.32\textwidth]{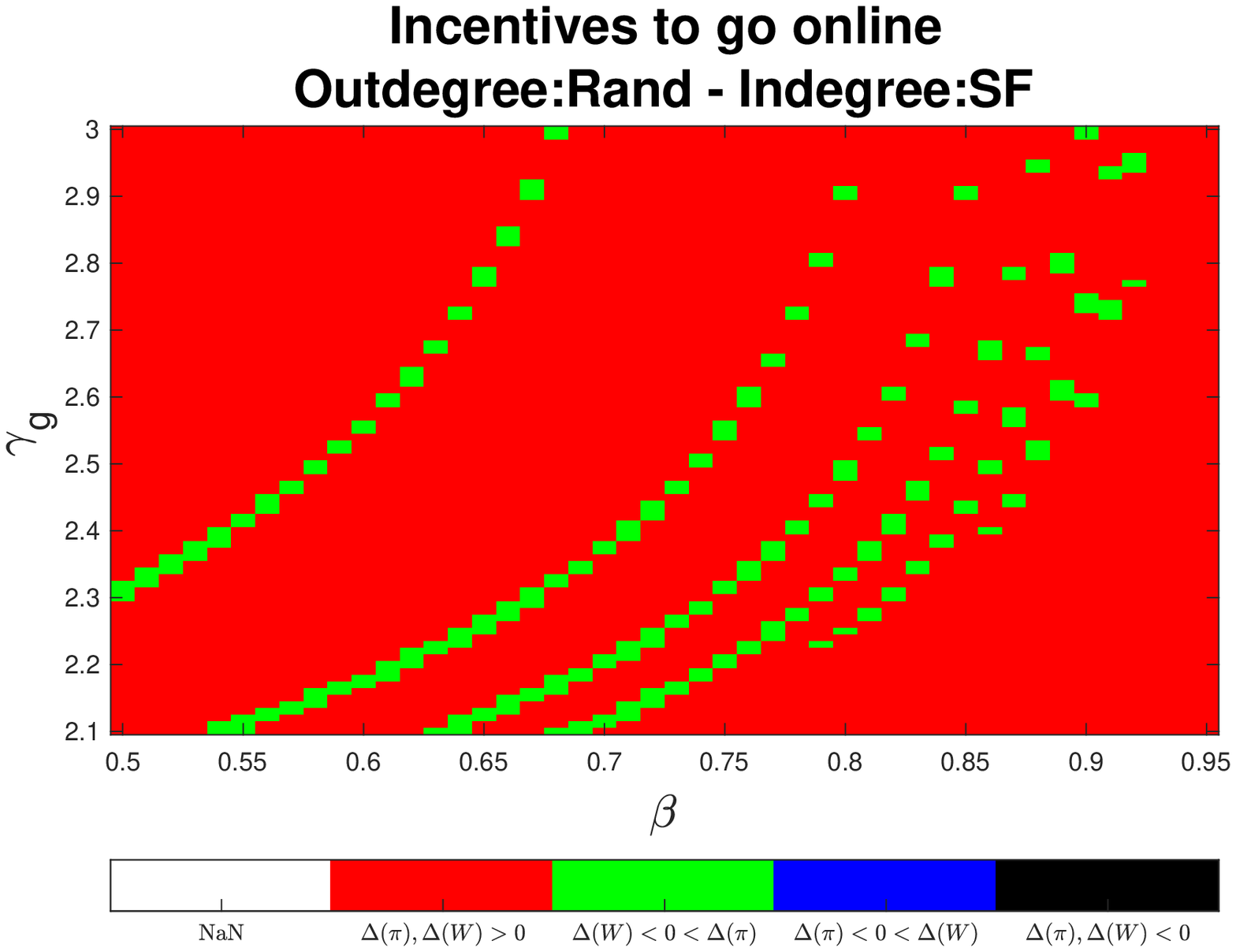}
         \includegraphics[width=.32\textwidth]{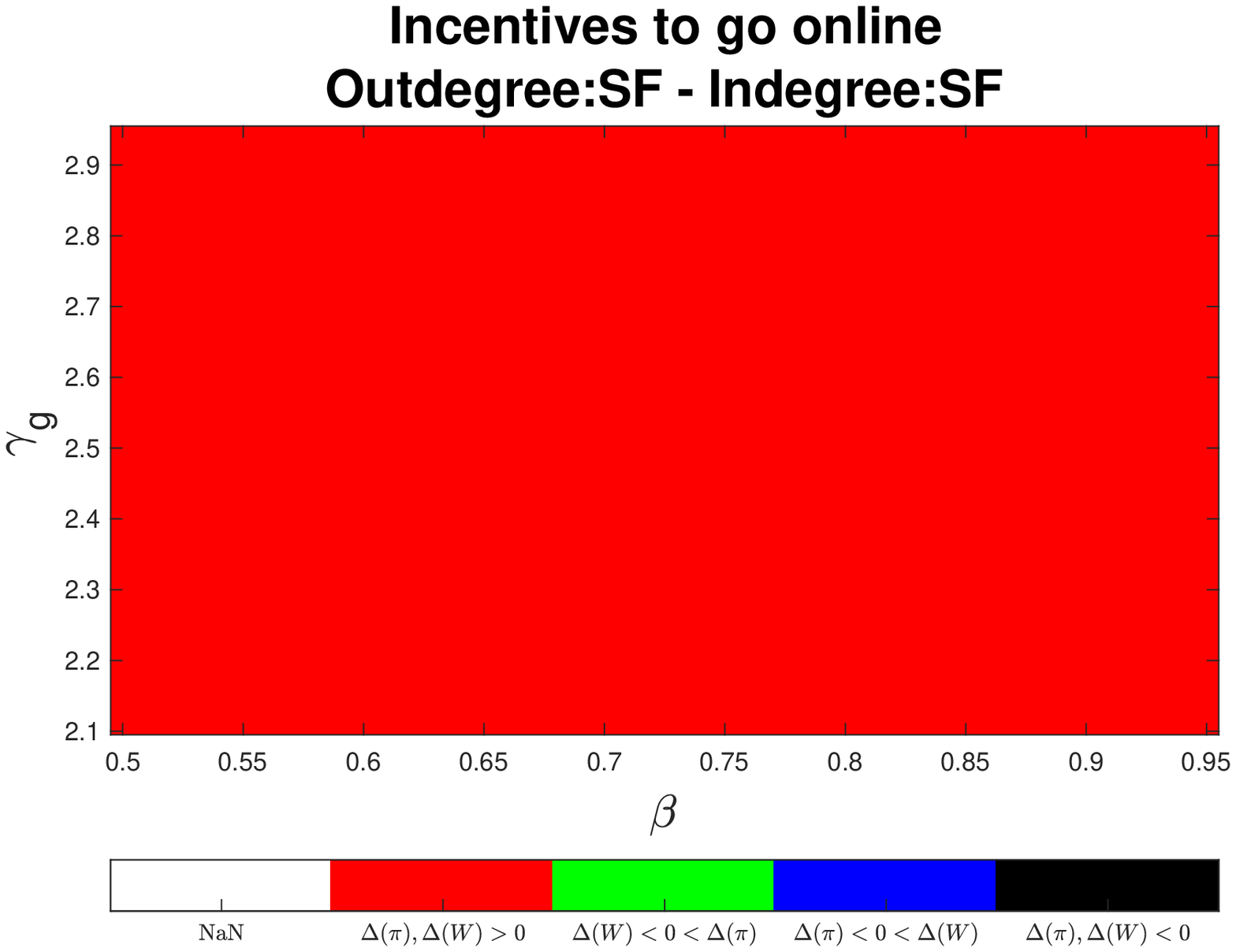}
        \caption{Effect of FOSD shift on profits and welfare. Left Panel: In-- and out--degree distributions  are random. Central Panel: Out--degree distribution  is random while the in--degree one  is scale-free. Right panel: In-- and out--degree distributions  are scale free.
For random in--degrees, we consider   FOSD shift of $2.5\cdot 10^{-2}$ for all combinations of starting $\lambda_g\in[0.05,0.95]$ in steps of $0.01$. For scale free in-degrees,  FOSD shift of $2.5\cdot 10^{-2}$ for all combinations of starting $\gamma_g\in[2.1,3]$ in steps of $0.01$. In both cases $\beta\in[0.5,0.95]$ in steps of $0.01$.  Number of agents: $10000$, $c=0.06$, $k_{max}=n-1$, $k_{min}=1$.  }
        \label{sc_fosd}
\end{figure}

When there are no hubs, the monopolist  balances the information diffusion  and the congestion inducing influencers  with a degree higher than the characteristic one to diffuse information. This is true for any given network density. As shown in Figure \ref{sc_fosd} (left panel) going online is surely profitable for the monopolist as far as  the original distribution of in--links is not too dense. When in--links are very dense, going online becomes less likely to be attractive than staying offline (blue and black areas in left panel). 

From the welfare standpoint, the picture is much richer and the observable absence of a  pattern essentially depends on discontinuous changes in the optimal $\underline{k}$, which  affect the amount of  information circulating in the network. When going online, the monopolist anticipates the stronger congestion. This may result in two reactions. The first one is to move  the optimal targeted $\underline{k}$ upwards so to partially offset the increase in congestion. This makes less information circulate in the denser network, with detrimental welfare consequences (green areas). The second one corresponds to cases in which the congestion effect prevails. Then, the monopolist does not want to move to the denser network and this can be in line with welfare (black) or not (blue).

Considering cases where the in--degree distribution presents hubs (central and right  panels of Figure \ref{sc_fosd}), moving to a denser network is always profit--increasing. This follows from two facts. Let us consider an originally \emph{sparse} network with few hubs attracting most of the inflow of information: these people are clearly  the main source of congestion for influencers. Once a denser network is faced, the increase in congestion is small relatively to the previous case, as these ``monopolizers'' of information still remain the main source of congestion.  Therefore, the monopolist will move online, since the positive effect on information diffusion associated to this movement  largely outweighs the additional cost of congestion. 

However, welfare--decreasing scenarios still occur. Indeed, the monopolist might want to reduce the proportion of senders to limit congestion and this can result in a reduction in the number of influenced people. This misalignment of incentives disappears when  hubs are present also in the out--degree distribution (Figure \ref{sc_fosd}, right panel). This result spurs from a reasoning similar to the one regarding influenced hubs. More precisely, influencing hubs are, for given network density,  the main source of informational spread. Therefore, a FOSD shift does not imply a strong change in the optimal behaviour of the monopolist, as those hubs still remain the most important diffusers. So the denser network always guarantees a higher spread of information at equilibrium.

\section{Conclusion}\label{conclusions}


Network--related referral strategies are an important demand--generating tool for today's companies.
The use of these strategies in many different markets has boomed thanks to online platforms and social media. Indeed, the latter enhance the outreach efforts of companies, giving them an effective and relatively cheap way to inform consumers about the products they sell.  Referral strategies look like win--win solutions when a seller faces a partially uninformed pool of consumers  as the former increases its demand incentivizing informed consumers with rewards, while uninformed people are made aware about possibly valuable products.


Many potential channels are available for companies interested in introducing referral programs. OSNs, sms, mails, instant messaging  are all means that can, together,  make a firm known by consumers in a very effective way. Nevertheless, frequently companies do not exploit all these possible channels but rather limit their strategy to some channels or use offline word--of--mouth. At first sight, a more penetrative strategy, maximizing information diffusion by fully exploiting all informed consumers' links seems to be trivially profit maximizing, especially since most firms condition influencers' rewards to the purchase of influenced people. 


In this paper, we show that the rationale for firms to limit the scope of potential activation of influencers' ties (or, to put it simpler,  to stay offline) is congestion, i.e., the fact that, in a network, non--informed people may receive information by multiple sources. As a result, the more numerous the buyers who become active influencers, the harder it is --  for each of them -- to get a bonus. Therefore, a seller has to consider the trade--off between demand size (informational spread)  and the marginal cost of attracting that demand (giving right incentives through bonuses). On this regard, network density is crucial and has two opposite effects. On the one hand,  it makes congestion stronger,  thus accruing the cost of providing incentives to communicate. On the other hand, it increases the information spread by influencers for fixed incentives.  When the first effect dominates, the seller prefers to  stay offline, somehow limiting the new demand attracted. Oppositely, if the latter demand effect prevails, maximizing the outreach becomes the dominant concern of the seller.


The attractiveness of bonuses is enhanced by the presence of few very connected influencers (\textit{hubs}), who are cheaper to incentivize and that guarantee a large information diffusion with little congestion. This is true for any given network density and drives the choice of the seller. 
When going online provides super--hub consumers that are not present offline, a company would always move to the denser network. In contrast, if online everybody is (homogeneously) more connected, then the seller might prefer to stay offline, where congestion is less severe.


This paper highlights the potential misalignment between seller's incentives and social welfare. 
Clearly, the latter is only concerned by information diffusion, as influenced people become able to buy something  whose existence they ignored. From the welfare viewpoint the bonuses paid by the monopolist are just a transfer and the thus the concerned trade--off  disappears. The prevalence of situations in which a misalignment emerges strongly depends on  the structure of the problem faced by the monopolist and in particular on the type of in-- and out--degree distributions. We provide conditions for this misalignment to appear in various situations with increasing levels of generality. When in--degrees are homogeneous, it is possible to find both cases of alignment (when both welfare and profits are increased by going online) and cases in which the monopolist prefers to stay offline at the welfare expenses. The prevalence of the two cases depends  essentially on the degree of homogeneity in influencers' degree distribution. 
When both influencers and influenced networks have homogeneous degree distributions, the seller always goes online and enhances welfare as this entails higher demand for any given level of communicators, without increasing the cost of word--of--mouth stimulation.  Differently, when the network of influencers is not necessarily homogeneous, it is always possible to find situations in which, if the cost of communication is sufficiently high, going online  induces too much congestion. 

Moreover, in presence of online super--hubs  the seller goes online and chooses to diffuse information only through web--influencers even at the expenses of reducing diffusion of information and welfare.  The crucial role of hubs is confirmed by an extensive analysis of numerical solutions on a wide class of networks.   When hubs are not important in neither networks any outcome can potentially emerge, i.e, we can have situations in which online/offline is preferred by welfare/seller, depending on parameters. The picture changes when offline there are few hubs among the influenced people (with out--degree distribution still symmetrically centered). These hubs are the main source of congestion. Hence,  their presence offline makes going online less consequential in terms of congestion thus making the move more appealing to the seller. In this case the seller always goes online, but this move can still reduce informational spread and thus welfare.  Finally, if disproportionately connected individuals are also present in the out--degree distribution, the incentives of the seller to mitigate congestion online are less substantial and this  aligns private and welfare incentives. 




The situation described by our paper is characterized by the presence of a monopolistic seller. However the setup can be extended to the case of a duopoly \emph{\`a la} Bertrand. In such setup the sellers compete for uninformed consumers receiving second--hand information through their network while informed consumers would enjoy the usual Bertrand's pressure on prices. As discussed in \ref{competition}, our paper endogenizes the size of the two pools of fully--informed and partially--informed consumers in \cite{burdett1983equilibrium}. They study a market in which firms are monopolists for partially--informed  consumers (who know only one good) and compete in prices for the fully--informed, creating price dispersion. In the light of our model, some initially uninformed consumers receive information about only one product and some others about both products. The density of the network becomes then crucial as a denser network entails a higher probability of being informed by both producers, creating an additional pressure on prices at the benefit of consumers.

\bibliographystyle{chicago}
\bibliography{BibFile}

\setcounter{section}{1}

\appendix

\global\long\def\thesection{Appendix \Alph{section}}
 \global\long\def\thesubsection{\Alph{section}.\arabic{subsection}}
 \setcounter{proposition}{0} \global\long\def\theproposition{\Alph{proposition}}
 \setcounter{definition}{0} \global\long\def\thedefinition{\Alph{definition}}

\section{}

\subsection{Proof of Proposition \ref{prop:all1/2}, page \pageref{prop:all1/2}} \label{proof_prop:all1/2}
Only agents with a sufficient out--degree will pass the information for given bonus. Therefore, by choosing $b$, the monopolist indirectly chooses the amount of informed consumers that will pass the information. The problem of the monopolist is equivalent to fixing an optimal level of $L^{*}$ or equivalently  $\rho^*$,\footnote{Notice that $L^{*}$ will then be $(1-p_{1}^{*})\rho^{*}$.}  and there is a one--to--one correspondence between $\rho^* \in [0,1]$ and a set of optimal $b^*$'s. 
So, the problem of the monopolist is to maximize the objective function:
\begin{equation}
 \label{eq:monopolist_problem2}
\pi(p_1,p_2,b)  = 
(1-\beta)  p_1 (1-p_1) + \beta  (p_2-b) (1-p_2) \Gamma(L) \ \ ,  
\end{equation}
and each solution to this problem is such that $p_1^*,p_2^* \in [0,1/2]$, and $b \in [0,p_2^*]$. \\

To maintain incentive compatibility for the consumers, $\rho$, in equilibrium, must satisfy  \eqref{threshold_k}, rewritten as
\begin{equation} \label{threshold_rho}
\Lambda(\rho) \phi (L,p_2) = \frac{ c}{b} \ \ .
\end{equation}
Both terms in the left--hand--side of \eqref{threshold_rho} are decreasing in $\rho$, so there is a unique $b$ satisfying each value of $\rho$.\footnote{%
Since $\Lambda(\rho) \phi ((1-p_1)\rho,p_2) $ is discontinuous, there are instead some values of $b$ for which no $\rho $ satisfies \eqref{threshold_rho} with equality. However, those values for $b$ are clearly dominated choices for the monopolist.} 

So, from  \eqref{threshold_rho}, the problem of the monopolist is equivalent to maximizing
\begin{equation}
 \label{eq:monopolist_problem3}
\pi(p_1,p_2,\rho)  = 
(1-\beta)  p_1 (1-p_1) + \beta  \left( p_2-   \frac{c}{ \Lambda(\rho) \phi ((1-p_1)\rho,p_2)  } \right) (1-p_2) \Gamma( (1-p_1) \rho) \ \ .
\end{equation}
Intuitively, an increase in the number of active links $L$ entails both an improvement in the flow of information - measured by $\Gamma (L)$ - and an increase in its congestion, i.e.~more competition for successfully passing the information - measured by the decrease in $\phi ((1-p_1)\rho,p_2)$.

The objective function \eqref{eq:monopolist_problem3} is discontinuous, but since it is bounded and has a finite number of jumps, it satisfies the extreme value theorem in each of the finite intervals. So, it has a solution which is the maximum of the maxima for each interval.

%
%
Let us take objective function \eqref{eq:monopolist_problem3}, and suppose that the monopolist chooses $\rho$ instead of $b$ (that will now depend endogenously on  the other parameters).
We consider optimal $p_1^*$ and $p_2^*$ for each given $\underline{k}$.

To compute optimal $p_2^*$, a necessary condition is to satisfy FOC:
\begin{eqnarray}
1-2p_2 & = & \frac{d}{d p_2} \left(  b (1-p_2) \right) \nonumber \\
& = &
-b + (1-p_2) \frac{\partial b}{\partial \phi} \frac{d \phi}{d p_2} \nonumber \\
& = & 
-b + \underbrace{(1-p_2) \left(  - \frac{c}{\phi^2 \underline{k}} \right)
\left(  - \sum_{k=1}^{\infty} g(k)  \left( \frac{1-(1-L)^k}{k L} \right) \right) }_{= \frac{b}{\phi} \cdot \phi} \ \ , \nonumber
\end{eqnarray}
but RHS is just $0$, so that $p_2^* = \frac{1}{2}$, independently on any other variable.\footnote{%
This does not even depend on the uniform distribution for the reservation price, which gives $(1-p_2)$, as this element is part of $\phi$ in both terms of the cancelation.}

To compute optimal $p_1^*$, we distinguish two cases.
First, optimal $\underline{k}=k_{f \max}+1$ and consequently $L=0$.
In this case there are no network effects and the optimal $p_1^*$ is trivially $\frac{1}{2}$.

Second case, $\underline{k} \leq k_{f \max}$ and $L>0$.
Now, a necessary condition   is to satisfy FOC:
\begin{eqnarray}
1-2p_1 &= & (1- p_2) \frac{\beta}{(1-\beta)} \left(  \Gamma(L) \frac{\partial b}{ \partial \phi} \frac{\partial \phi}{ \partial L} \frac{d L }{d p_1} -
 \left( p_2 - b \right) \frac{\partial \Gamma}{\partial L} \frac{d L }{d p_1} \right)  \nonumber \\
& = &
 - \frac{1}{2} \frac{\beta}{(1-\beta)} \frac{d L }{d p_1} 
\left(  \left( \frac{1}{2} - b \right) \frac{\partial \Gamma}{\partial L}  - \Gamma(L) \frac{\partial b}{ \partial \phi} \frac{\partial \phi}{ \partial L}    \right).
 \ \ 
\label{eq:optimal_p1}
\end{eqnarray}

The piece 
\[
\left( \frac{1}{2} - b \right) \frac{\partial \Gamma}{\partial L}  - \Gamma(L) \frac{\partial b}{ \partial \phi} \frac{\partial \phi}{ \partial L}    
\]
is actually equal to $\frac{d \Pi}{d  L}$: the marginal effect of $L$ on profits.

If the optimal $\rho^* $ is not a point of discontinuity for $\Lambda(\rho)$ (and then consumers at $\underline{k}$ play different strategies), then $\frac{d \Pi}{d  L} = 0$, and then $p_1^{*}=1/2$. \\
If instead  $\rho $ is a point of discontinuity for $\Lambda(\rho)$  (and then all consumers at $\underline{k}$ pass the information),  then $\frac{d \Pi}{d  L}$ is not necessarily null.
However, consider that the objective function \eqref{eq:monopolist_problem3} is lower semicontinuous, because the only non--continuous element in it is $\Lambda(\rho)$, which is lower semicontinuous.
Because of that,
 $\frac{d \Pi}{d  L}$ cannot be negative in some point of discontinuity, otherwise $\rho^*-\epsilon$ would be a better choice than $\rho^*$. 
Looking at Equation \eqref{eq:optimal_p1}, since $\frac{d L }{d p_1}<0$, we obtain that $p_1^{*} \leq 1/2$.  
\eproof
%

\subsection{General distribution}\label{GENERALf}

Consider reservation values to be distributed according to a continuously differentiable  distribution with c.d.f.  $H(\cdot)$ and density $h(\cdot)$. In order to guarantee concavity in prices, the distribution is assumed  to have monotonically increasing hazard rate, i.e., $\frac{\partial\left( \frac{v}{1-H(v)}\right)}{\partial v}>0$.
The problem of the monopolist then reads: 
\begin{equation}
\pi(p_1,p_2,b)  = 
(1-\beta)  p_1 (1-H(p_1)) + \beta  (p_2-b) (1-H(p_2)) \Gamma(L), \nonumber
\end{equation}
 where 
 
 \begin{equation}
\Gamma(L) = \sum_{k=1}^{\infty} g(k) \left(  1-(1-L)^k \right) = 1 - \sum_{k=1}^{\infty} g(k) (1-L)^k \ \ . \nonumber
\end{equation}

\begin{eqnarray}
\phi (L,p_2) &  = & (1-H(p_2))  \sum_{k=1}^{\infty} g(k)  \left( 1 - \sum_{t=0}^{k-1} {k-1 \choose t} L^t (1-L)^{k-1-t} \frac{t}{t+1} \right) \nonumber
\\
& = & (1-H(p_2))  \sum_{k=1}^{\infty} g(k)  \left( \frac{1-(1-L)^k}{k L} \right).  \nonumber
\end{eqnarray}
 and 
 
\begin{equation}
L=[1-F(p_{1})]\rho \nonumber
\end{equation}

The condition in \eqref{threshold_rho} is again needed to guarantee communication to be incentive compatible. As  in Appendix \ref{proof_prop:all1/2}, we consider optimal $p_1^*$ and $p_2^*$ for each given $\underline{k}$. 
To compute optimal $p_2^*$, a necessary condition is to satisfy FOC:
\begin{eqnarray}
1-H(p_2)-h(p_{2}) & = & \frac{d}{d p_2} \left(  b (1-H(p_2)) \right) \nonumber \\
& = &
-h(p_{2})b + (1-H(p_2)) \frac{\partial b}{\partial \phi} \frac{d \phi}{d p_2} \nonumber \\
& = & 
-h(p_{2})b + \underbrace{(1-H(p_2)) \left(  - \frac{c}{\phi^2 \underline{k}} \right)
\left(  - h(p_{2}) \sum_{k=1}^{\infty} g(k)  \left( \frac{1-(1-L)^k}{k L} \right) \right) }_{= h(p_{2})\frac{b}{\phi} \cdot \phi} \ \ , \nonumber
\end{eqnarray}
 so that $p_{2}^{*}$ is equal to the monopoly price $p_{m} = \frac{1-F(p_{2}^{*})}{f(p_{2}^{*})}$.

Clearly, if optimal $\underline{k}=k_{f \max}+1$, then   $p_1^*=p_{m}$. When instead $\underline{k} \leq k_{f \max}$ and $L>0$, a necessary condition   is to satisfy FOC:
\begin{eqnarray}
1-H(p_1)-h(p_{1})p_{1} &= & (1- H(p_2)) \frac{\beta}{(1-\beta)} \left(  \Gamma(L) \frac{\partial b}{ \partial \phi} \frac{\partial \phi}{ \partial L} \frac{d L }{d p_1} -
 \left( p_2 - b \right) \frac{\partial \Gamma}{\partial L} \frac{d L }{d p_1} \right)  \nonumber \\
& = &
 -(1- H(p_m)) \frac{\beta}{(1-\beta)} \frac{d L }{d p_1} 
\left(  \left( p_{m} - b \right) \frac{\partial \Gamma}{\partial L}  - \Gamma(L) \frac{\partial b}{ \partial \phi} \frac{\partial \phi}{ \partial L}    \right)\nonumber \\
& = &
 (1- H(p_m)) \frac{\beta}{(1-\beta)}h(p_{1})\rho
\left(  \left( p_{m} - b \right) \frac{\partial \Gamma}{\partial L}  - \Gamma(L) \frac{\partial b}{ \partial \phi} \frac{\partial \phi}{ \partial L}    \right). \nonumber
 \ \ 
\end{eqnarray}

The factor  
\[
\left( p_{m} - b \right) \frac{\partial \Gamma}{\partial L}  - \Gamma(L) \frac{\partial b}{ \partial \phi} \frac{\partial \phi}{ \partial L}    
\]
is the marginal effect of $L$ on profits and it is positive  only when $\rho $ is a point of discontinuity for $\Lambda(\rho)$, otherwise it is null.
In points of discontinuity,  the monotone hazard rate assumption  implies  that $p_1^{*} < p_{m}$. So the results are in line with the one in Proposition   \ref{prop:all1/2}.

\subsection{Numerical analysis of $p_1$}
\label{numerical_analysis}

The price $p_1^*$ needs further clarifications.  It turns out that it is generically lower than $1/2$. This comes from the fact that $k$ is discrete, so that an  increase in the cost of communication $c$ would generate a discrete jump in the optimal $\underline{k}^*$, and the marginal effect of $L$ on the profit will be zero in interior points and can be higher than zero when $\rho^{*}$ is a discontinuity point of $\Lambda(\rho^*)$. For these reasons, when $\underline{k}^{*}$  moves in response to a marginal increase in $\rho$ the price $p_1$ suddenly jumps below $1/2$. 


Figure \ref{c_on_k} shows that the minimal number of ties $\underline{k}^*$ needed to communicate increases in sharp jumps  as $c$ increases. While such behavior is present for all network sizes (and indeed it never disappears), a unitary jump in $\underline{k}^*$  becomes progressively negligible with respect to  the  total network size. 
Moreover the $\Delta c$ from one jump to the next becomes smaller, an is negligible as $N \rightarrow \infty$.

The counterparts figures of Figure  \ref{c_on_k}, which map $p_1^*$ in function of $c$ for different network sizes, show how the price tends to $1/2$. Jumps are only present in points of discontinuity and vanish as network size increases. Indeed, jumps in the price are present in the order of $10^{-13}$ when the network size is $500$ (Left Panel)  and they become negligible (of the order of $10^{-133}$) when the network size is $5000$ (Right Panel), as it can be observed in Figure \ref{c_on_p1}.\footnote{ For networks of 50000 individuals jumps size falls below the maximum numerical precision of $10^{-308}$.} Let us notice that, by comparison with the markets sizes faced by most companies doing referral marketing, these numbers are very small.

 \begin{figure}[h!]
 \centering
        \includegraphics[width=.32\textwidth]{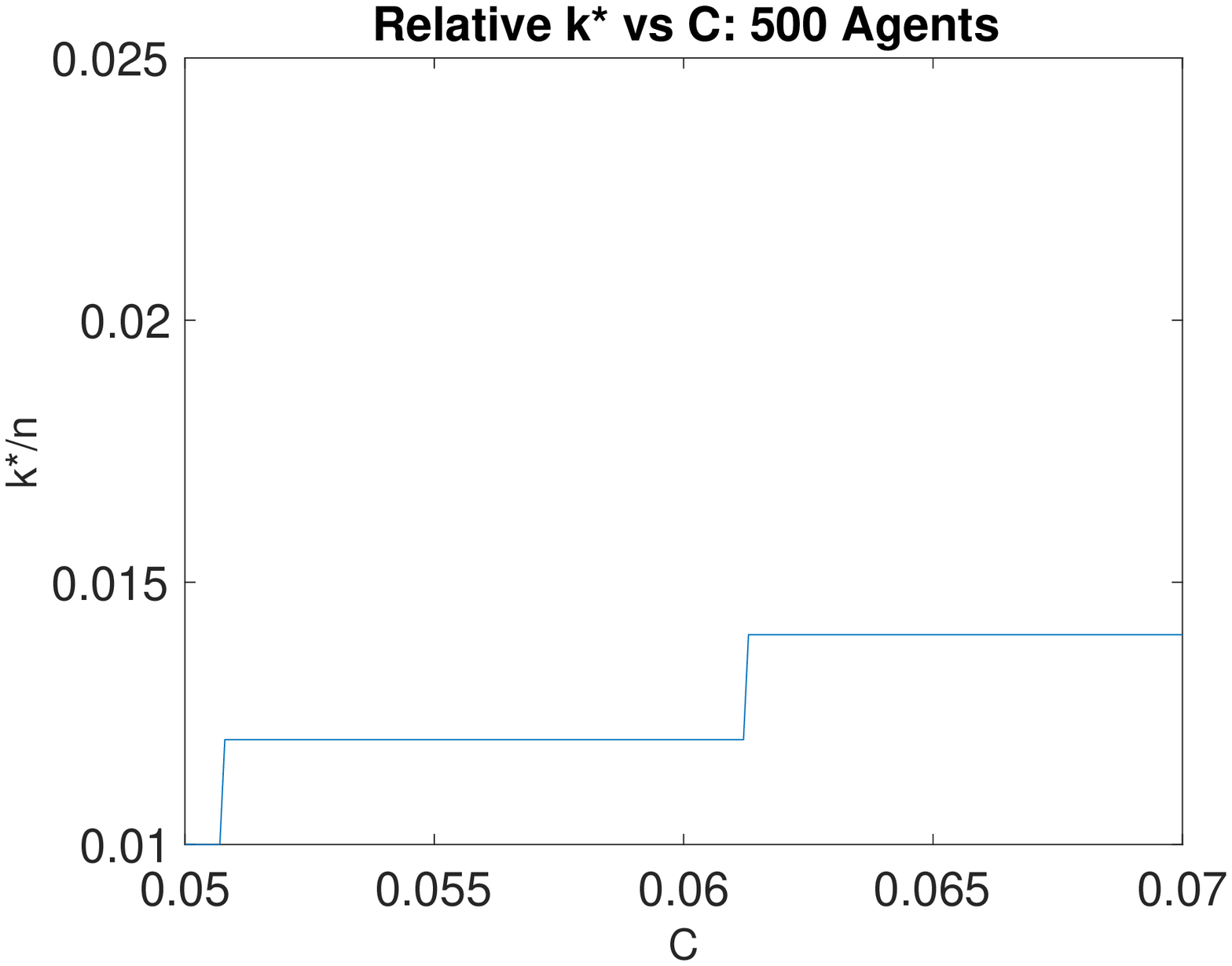}
        \includegraphics[width=.32\textwidth]{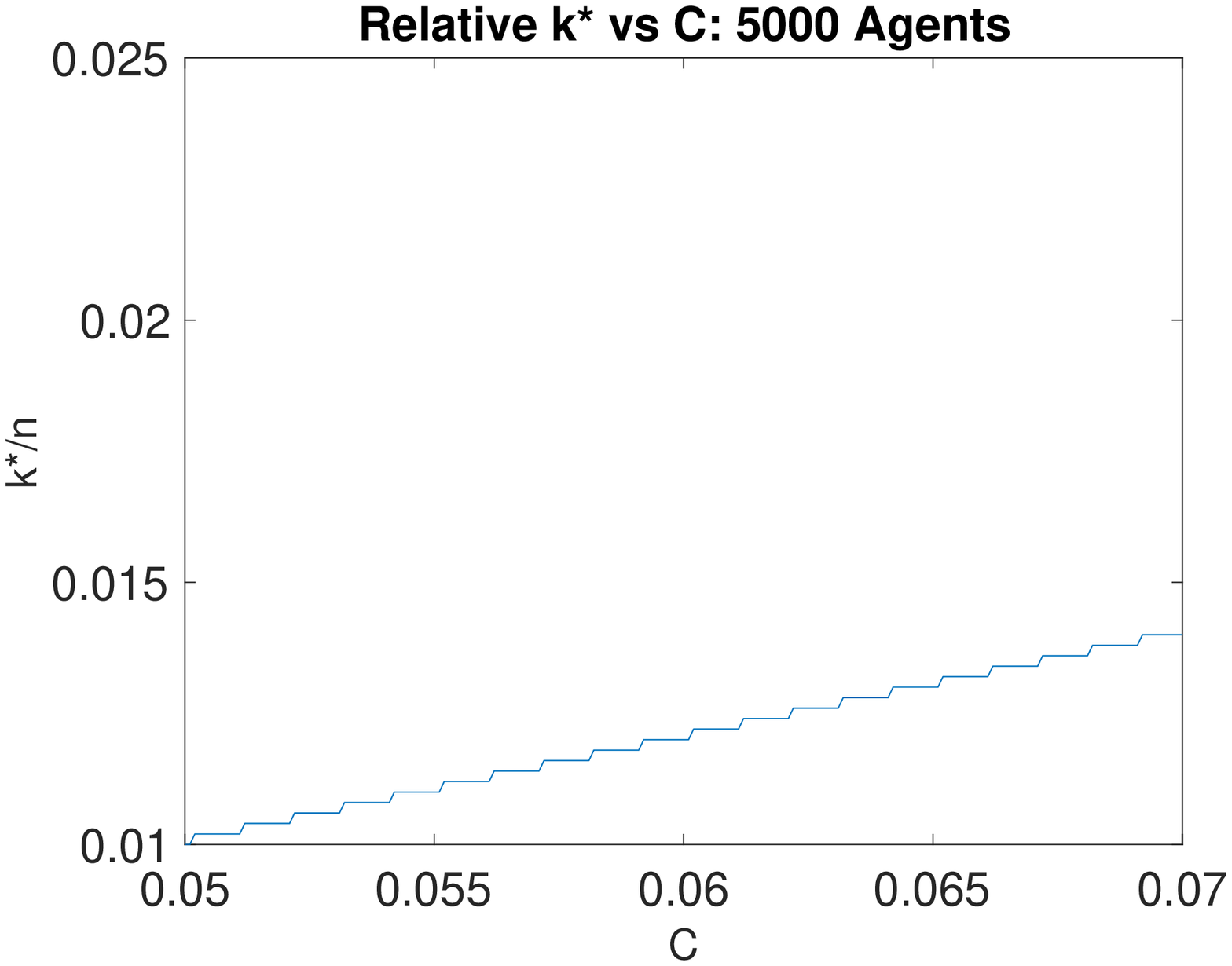}
        \caption{Effect of $c$ on $\underline{k}^*$. Results are reported  for a random network with probability of link equal to $0.1$ and  $\beta=0.3$. Left panel: $500$ agents. Right Panel: $5000$ agents.}
        \label{c_on_k}
\end{figure}

\newpage

\section{Other formal proofs}

\subsection{Proof of Proposition \ref{prop:misalignment}, page \pageref{prop:misalignment}} \label{proof_misalignment}

Since only the choice of   $\underline{k}$ is concerned in optimizing \eqref{eq:monopolist_problem3},  we  refer only to  the profit the monopolist makes on second--period consumers, expressing $b$ in terms of the other variables:  
\begin{equation}
\label{profit_piece}
\pi_2(1/2, \underline{k} )  = \frac{\beta}{2} \left( \frac{1}{2}-\frac{c}{\phi \underline{k}} \right)\Gamma(L)\ \ .
\end{equation}
The monopolist would maximize welfare if $\underline{k}=k_{f \min}$ and then $L=1/2$. 
In this case 
\[
\phi(L) = (1-p_2)  \sum_{k=1}^{\infty} g(k)  \left( \frac{1-(1-L)^k}{k L} \right) = \sum_{k=1}^{\infty} g(k) \left( \frac{1-\frac{1}{2}^k}{k} \right) \leq   E_k [ 1/k ] \ \ .
\]
However, if $\phi(L) <  \frac{2c}{k_{f \min}}$, then $\frac{1}{2}-\frac{c}{\phi k_{f \min}} <0 $, and so the part of profits from (\ref{profit_piece}) is negative. 
If instead the monopolist chooses $\underline{k} > k_{f \max}$, then $L=0$, and the part of profits from (\ref{profit_piece}) is null.
So,  $\underline{k}=k_{f \min}$ cannot be an optimal choice for the monopolist, because it is dominated by $\underline{k}=k_{f \max}$.
\eproof

\subsection{Proof of Proposition \ref{FOSDWelfare}, page \pageref{FOSDWelfare}} \label{proof_FOSDWelfare}

Since the threshold $\underline{k}$ remains fixed, a  FOSD shift from  $f(k)$ to $f'(k)$   makes 
\[
L = \frac{1}{2} \sum_{k=\underline{k}}^{\infty} f(k) \leq \sum_{k=\underline{k}}^{\infty} f'(k) = L'
\ \ .
\]
Therefore, each $1-(1-L)^k$ increases. 
In the expression of $\Gamma(L')$, as given by Equation \eqref{eq:Gamma},
$1-(1-L)^k$ is also an increasing function in $k$.
So, given that also $g'(k)$ FOS--dominates $g(k)$, we can use Definition \ref{def_fosd} and obtain
\[
\Gamma(L) = \sum_{k=1}^{\infty} g(k) \left(  1-(1-L)^k \right) \leq \sum_{k=1}^{\infty} g(k) \left(  1-(1-L')^k \right) \leq \sum_{k=1}^{\infty} g'(k) \left(  1-(1-L')^k \right) = \Gamma (L')
.\]
So, as 
$ \Gamma(L) \leq \Gamma (L') $,
welfare from Equation \eqref{eq_welfare} increases.
\eproof

\subsection{Proof of Proposition \ref{prop:FOSD_homogeneous}, page \pageref{prop:FOSD_homogeneous}}\label{FOSD_homogeneous}

Let us consider a network where all consumers have  in--degree equal to $k_g$ and out--degree equal to $k_f$.  From the consistency condition in \eqref{eq:consistency} it must hold that $(1-\beta)k_f=\beta k_g$.  Plugging into  \eqref{eq:L} and, in turns, into \eqref{eq:Gamma} and \eqref{eq:phi}, we get:

\begin{equation}
\label{system_hom_degree}
\begin{array}{c}
L=\rho/2,\\
\Gamma(\rho)=1-(1-\rho/2)^{k_g},\\
\phi(\rho/2, k_g)=\frac{1}{\rho k_g}\left(1-(1-\rho/2)^{k_g}\right)=\frac{1}{\rho k_g}\Gamma(\rho).
\end{array}
\end{equation}
Using consistency condition \eqref{eq:consistency}, the profit in Equation \eqref{profit_piece} becomes either null (if the monopolist does not use word of mouth)\footnote{ This will be the case whenever $\frac{c}{\phi(\rho)k_{f}}<1/2$ for all $\rho\in(0,1)$, so that profits made on uninformed consumers would  be negative.  } or: 
\begin{equation}
 \frac{\beta}{2} \left( \frac{1}{2}-  \frac{1-\beta}{\beta} \frac{\rho c }{\Gamma(\rho)} \right)\Gamma(\rho)
= \frac{\beta}{4} \Gamma(\rho) -  \rho (1-\beta) c   \ \ . 
\label{simple_objective}
\end{equation}
Notice that Equation \eqref{simple_objective} has a concave increasing part and a linearly decreasing one.
A FOSD shift would affect the former but not the latter. On the one hand,  in a denser network the function $\Gamma'(\rho)$ maps $\rho$ into higher values than the original $\Gamma(\rho)$. On the other hand, the FOSD shift does not entail any effect on the linear part, which depends only on the slope $(1-\beta) c$ that is not affected by network density. Therefore, for any $\rho$ the after--FOSD profit will always be above than the pre--FOSD one. 
Both welfare and monopolist's profit increase.
\eproof

\subsection{Proof of Proposition \ref{prop:HOMOIN}, page \pageref{prop:HOMOIN} }\label{HOMOIN}

Let us consider a network where all consumers have  in--degree equal to $k_g$.  
Let us first consider what would be the profit of the monopolist chosing a certain $\rho$.
Plugging into  \eqref{eq:L} and, in turn, into \eqref{eq:Gamma} and \eqref{eq:phi}, we get again the system from 
\eqref{system_hom_degree}.
Now, however, $\underline{k}=\Lambda(\rho)$ will vary with $\rho$, and we have that the profit in Equation \eqref{profit_piece} becomes either null (if the monopolist does not use word of mouth) or:
\begin{equation}
 \frac{\beta}{2} \left( \frac{1}{2}-  \frac{ \rho k_g}{\Lambda(\rho)} \frac{c}{\Gamma(\rho)} \right)\Gamma(\rho)
= \frac{\beta}{4} \Gamma(\rho) -  c \frac{\rho k_g}{\Lambda(\rho)}  \ \ . \nonumber
\end{equation}

Now suppose to fix both $\underline{k}$ and the fraction $q$ of consumers with degree $\underline{k}$ that pass the information in equilibrium.
We have that:
\[
\rho = q f(\underline{k}) \cdot \underline{k} + \sum_{k=\underline{k}+1}^{k_{f \max} } f(k) \cdot k \ \ .
\]
Any couple $(\underline{k},q)$ will be feasible also going to a denser network, and fixing them we will have $\rho'>\rho$.\\
For any couple $(\underline{k},q)$, the profit in Equation \eqref{profit_piece} becomes: 
\begin{eqnarray}
\pi (\underline{k},q) = \frac{\beta}{4} \Gamma(\rho)-    c \frac{k_g\rho}{\underline{k}} \ \ .
\label{profit_piece_lattice}
\end{eqnarray}

Now consider the case in which $\underline{k}=k_{f \min}$ (minimal cost for the monopolist to provide incentives to communication) and $\Gamma\rightarrow1$ (maximal informational spread). A sufficient condition for each couple $(\underline{k},q)$ to provide negative profits is that in this best scenario gives negative profits, or formally:\footnote{Notice that, in this simple case, we would have $\rho=f(k_{f \max})$, i.e. only more connected influencers pass the information and maximal informational spread is guaranteed.}  
 
\[
c > \frac{\beta}{4} \frac{k_{f \max}}{k_g f(k_{f \max})} \ \ ,
\]
that using condition \eqref{eq:consistency}, becomes

\[
c > \frac{\beta^{2}}{4} \frac{k_{f \max}}{(1-\beta)E_f (k) f(k_{f \max})} \ \ ,
\]

Going to a denser network we can get $\frac{k'_{f \max}}{ E_f (k') }$ as close to $1$ as we want, so that $f(k_{f \max})$ as well goes to $1$. Therefore,  
$c > \frac{\beta^2}{4(1-\beta) }$ becomes a sufficient condition for the existence of a denser network that is not profitable for the monopolist.
Notice that if $k_{f \max}$ was high enough in the original network, then it could have been possible to find a couple $(\underline{k},q)$  such that \eqref{profit_piece_lattice} was positive in the original sparser network, but negative in the denser network.
\eproof

\subsection{Proof of Proposition \ref{prop:INHOMOWELF}, page \pageref{prop:INHOMOWELF}}\label{INHOMOWELF}

From Equations \eqref{eq_welfare} and \eqref{eq:Gamma}, if the monopolist chooses some $\rho$, welfare depends only on 
\[
\Gamma(k_g) = 1-(1-L (\rho^{*}))^{k_g} \ \ .
\]
When $k_g$ increases to some $k'_g$, $\rho^*$ may move to some $\rho'$. In order to look for sufficient conditions, consider the limit situation where 
$\rho^*=1$, so that $L=1/2$. Plugging into $\Gamma$ we get that going online is welfare enhancing if:

\begin{eqnarray}
\Gamma'>\Gamma
& \Leftrightarrow &
\left(\frac{1}{2}\right)^{k_g}>\left(1-\frac{\rho'}{2}\right)^{k'_g}  \nonumber  \\
& \Leftrightarrow & k'_g <-\frac{\log{(2)}}{\log{[1-\rho'/2]}}k_g  \nonumber  \\
& \Leftrightarrow & 
\frac{k_g}{k'_g}<
1-\frac{\log{[2-\rho']}}{\log{(2)}}  \ \ .
\label{SuffCondFOSDwelfareenh}
\end{eqnarray}

Condition in \eqref{SuffCondFOSDwelfareenh} is sufficient for the FOSD shift to be welfare enhancing.
%
%
\eproof

%
%

\subsection{Proof of Proposition \ref{prop:general} page \pageref{prop:general}}\label{proof:general}

Suppose to create \emph{super hubs} in the $f$ distribution, and to increase uniformly the $g$ distribution.
Then, to respect the network constraint, we give degree $\bar{k}>k_{f \max}$ to the hubs, and add $m$ links to the uninformed consumers, such that
\begin{eqnarray}
(1-f(\bar{k}) \left( \sum f(k) k \right) + f(\bar{k} ) \bar k & = & \frac{\beta}{1-\beta} \left( \sum g(k) (k+m) \right)  \nonumber \\
& = & \left( \sum f(k) k \right) + m \frac{\beta}{1-\beta} 
\ \ , \label{consistencyFOSD}
\end{eqnarray}
which implies
\[
f(\bar{k} ) \bar{ k} =  m \frac{\beta}{1-\beta} +  f(\bar{k} ) \left( \sum f(k) k \right)  \ \ .
\]
It is important to consider that all these expressions for $m$ make sense even if  the latter is a real non--integer number, and we assume to give the integer part $\lfloor m \rfloor$ to each uninformed consumer and the remaining real part $m-\lfloor m \rfloor$ is a probability for each of them to receive the $\lfloor m \rfloor+1^{th}$ link. That is because Equation \eqref{consistencyFOSD} is linear in $m$.

\bigskip
%
\noindent {\bf Step 1: for any $\epsilon>0$, we can get $\Gamma - \epsilon < \Gamma' < \Gamma$.} \\
The old $\Gamma$ was
\[
\Gamma= 1- \left( \sum g(k) \left(  1-\frac{1}{2} \right)^k \right)  \ \ .
\]
The new $\Gamma'$ will become
\[
\Gamma'= 1- \left( \sum g(k) \left(  1-\frac{\rho'}{2} \right)^k \right) \left(  1-\frac{\rho'}{2} \right)^m \ \ .
\]

Since $\lim_{m\rightarrow 0} \Gamma'=0$ and $\lim_{m\rightarrow \infty} \Gamma'=1$, we can pick arbitrarily any value for $\Gamma'$, keeping also $\bar{k}$ free to move as high as wanted (since we have still freedom on $f(\bar{k})$).
So, for each $\epsilon>0$, we can take 
\[
\Gamma-\epsilon < \Gamma' < \Gamma \ \  ,
\]
and are still free to pick any $\bar{k}$.

\bigskip

\noindent {\bf Step 2: we can improve on the payoff from the old network.} \\
Now, in the monopolist problem \eqref{eq:monopolist_problem4} the objective function can be reduced to
\[
\left( \frac{1}{2} -   \frac{c}{ \Lambda(\rho) \phi \left( \frac{\rho}{2}, \frac{1}{2} \right)  } \right)  \Gamma(   \rho / 2 ) \ \  ,
\]
and in this case we compare
\[
\left( \frac{1}{2} -   \frac{c}{ \Lambda(\rho^*) \phi \left( \frac{\rho^*}{2}, \frac{1}{2} \right)  } \right)  \Gamma  
\mbox{ \ \ with \ \ } 
\left( \frac{1}{2} -   \frac{c}{ \bar{k} \phi' \left( \frac{\rho'}{2} , \frac{1}{2} \right) } \right)  (\Gamma  - \epsilon ) \ \  .
\]
As we can take $f(\bar{k} ) \rightarrow 0 $ and $ \bar{k} \rightarrow \infty $, let us show that we can achieve 
\[
\lim_{\bar{k} \rightarrow \infty} \bar{k} \phi' \left( \frac{\rho'}{2} , \frac{1}{2} \right) = + \infty \ \ .
\]
This is always possible if we consider that $\phi' \left( \frac{\rho'}{2} , \frac{1}{2} \right)$ is bounded below by
\[
E' (1/k) = \sum g'(k)/k =  (1-f(\bar{k}) ) E (1/k)  + \left(   m \frac{\beta}{1-\beta} +  f(\bar{k} ) \left( \sum f(k) k \right) \right)  \frac{1}{\bar{k}^2} \ \ .
\]
Since $\Gamma<1$, we can set $\epsilon <  \frac{c}{ k_{f \min} \phi \left( \frac{1}{2}, \frac{1}{2} \right)}$, and we can set a  $ \bar{k}$ high enough so that 
\[
\left( \frac{1}{2} -   \frac{c}{ \Lambda(\rho^*) \phi \left( \frac{\rho^*}{2}, \frac{1}{2} \right)  } \right)  \Gamma
<
\left( \frac{1}{2} -   \frac{c}{ \bar{k} \phi' \left( \frac{\rho'}{2} , \frac{1}{2} \right) } \right)  ( \Gamma  - \epsilon ) \ \  .
\]
\bigskip

\noindent {\bf Step 3: in the new network, the monopolist will not choose an $\underline{k}<\bar{k}$ .} \\
Consider again the monopolist problem \eqref{eq:monopolist_problem4}, where the objective function can be reduced to
\[
\left( \frac{1}{2} -   \frac{c}{ \Lambda(\rho) \phi \left( \frac{\rho}{2}, \frac{1}{2} \right)  } \right)  \Gamma(   \rho / 2 ) \ \  .
\]
We have seen that we can set 
\[
\lim_{\bar{k} \rightarrow \infty} \bar{k} \phi' \left( \frac{\rho'}{2} , \frac{1}{2} \right) = + \infty \ \ .
\]
It is not difficult to see that all the levels of $\underline{k}$ that were feasible in the old network will be feasible also in the new network, but at a higher cost.
That is because $\underline{k}$ will be the same but expected gains from $\phi$ will  be decreased by the congestion created by the new super hubs.
So, suppose that the monopolist chooses some $\rho''> \rho'$, then we can still maintain the following inequalities:
\[
\left( \frac{1}{2} -   \frac{c}{ F(\rho'') \phi \left( \frac{\rho''}{2}, \frac{1}{2} \right)  } \right)  
<
\left( \frac{1}{2} -   \frac{c}{ \Lambda(\rho^*) \phi \left( \frac{\rho^*}{2}, \frac{1}{2} \right)  } \right) 
< 
\left( \frac{1}{2} -   \frac{c}{ \bar{k} \phi' \left( \frac{\rho'}{2} , \frac{1}{2} \right) } \right)  ( 1 - \delta - \epsilon ) \ \  .
\]
The last step is guaranteed by the assumption that $\delta  <  \frac{2 c}{ \Lambda(\rho^*) \phi \left( \frac{\rho^*}{2}, \frac{1}{2} \right)}$, and by the fact that we can achieve 
$\epsilon \rightarrow 0$ and 
\[
\lim_{\bar{k} \rightarrow \infty} \bar{k} \phi' \left( \frac{\rho'}{2} , \frac{1}{2} \right) = + \infty \ \ .
\]
So for any of the feasible $\underline{k}$ in the old network, we can choose $\epsilon$, $\bar{k}$ and $m$, so that the monopolist can do better in the new network choosing to do word--of--mouth only with  the superhubs.
Since, the set of all feasible $\underline{k}$'s in the old network is finite, we can satisfy the conditions of this last step of the proof.
\eproof

\section{Competition} \label{competition}

It is natural to ask what would happen, in the context of our model, if we allowed for competition.
We provide here a simple example to show that in this case the analysis and the results would fall in the domain of the already well--established literature on Bertrand competition with partially informed agents (see e.g.~\citealt{burdett1983equilibrium}).

\bigskip

Consider a situation in which two symmetric firms, $A$ and $B$, compete \emph{\`a la} Bertrand in both markets: the one for informed and the one for uninformed.
To make the example very simple consider also a situation in which $c=0$, so that there is an equilibrium in which  every informed consumer passes the information even if with $b=0$. \\
In this equilibrium the two firms would compete \emph{\`a la} Bertrand in the market of informed consumers, so that we can safely assume that each of the two firms get the same fraction of consumers: $L_A=L_B=\frac{1}{2}$.
Each firm then will reach a fraction $\Gamma(1/2) > 1/2$ of the uninformed consumers, and every uninformed consumer will be reached by news of at least one firm.

Because of symmetry (since we assume no degree correlation in the network structure), this will result in a fraction $1-\Gamma(1/2)$ of consumers who observe only   firm $A$,  $1-\Gamma(1/2)$ who observe only firm $B$, and the remaining $2 \Gamma(L) - 1$ who observe both prices.
The market for uninformed consumers is then exactly the environment studied in \cite{burdett1983equilibrium}, and we know from there that expected prices will be higher when the fraction of consumers who observe only one price, which is $2 - 2 \Gamma(L)$, is higher.\footnote{%
Note that in that literature \emph{informed} and \emph{uninformed} consumers are those observing two prices or one price, respectively. In our setting we are calling \emph{uninformed} all the consumers in this market, because they all get second--hand information about the existence of the product.}
As intuition would suggest, if $g(k)$ gets denser, then $\Gamma(1/2)$ increases, so that consumers become \emph{more informed} and prices go down, harming the monopolist and benefiting consumers.
Since in this case all consumers would buy the product anyway, welfare remains unaffected.

\end{document}